\documentclass[a4paper,fleqn]{cas-sc}

\usepackage[sort&compress,numbers]{natbib}
\usepackage{float}
\usepackage{graphicx}
\usepackage{tabularx}
\usepackage{subcaption}
\usepackage{caption}
\def\tsc#1{\csdef{#1}{\textsc{\lowercase{#1}}\xspace}}
\tsc{WGM}
\tsc{QE}
\tsc{EP}
\tsc{PMS}
\tsc{BEC}
\tsc{DE}
\newtheorem{theorem}{Theorem}

\newdefinition{rmk}{Remark}
\newdefinition{assump}{Assumption}
\newdefinition{pf}{Proof}

\begin{document}
\let\WriteBookmarks\relax
\def\floatpagepagefraction{1}
\def\textpagefraction{.001}
\shorttitle{}
\shortauthors{J. Li et~al.}
\let\printorcid\relax

\title [mode = title]{On Delay-robustness of Extremum Seeking of Nonlinear Static Maps with Small Disturbance}                      
\tnotemark[1]

\tnotetext[1]{This work is supported by National Natural Science Foundation of China (Grant No. 62303410, Science Fund Program for Distinguished Young Scholars (Overseas)), Zhejiang Provincial Natural Science Foundation of China (Grant No. LQ23F030014).}


\author[1,2]{Jianzhong\ Li}[
]
\ead{jianzhong_li626@163.com}

\credit{Investigation, Methodology, Writing – original draft, Validation}

\affiliation[1]{organization={School of Information and Control Engineering, Southwest University of Science and Technology},
                city={Mianyang},
                postcode={621010}, 
                country={China}}

\author[2]{Yang\ Zhu}[
]
\cormark[1]
\ead{zhuyang88@zju.edu.cn}
\credit{Methodology, Writing – review and editing, Supervision, Funding acquisition}

\author[2]{Hongye\ Su}[%
   ]
\ead{hysu69@zju.edu.cn}
\credit{Methodology, Supervision, Funding acquisition}

\affiliation[2]{organization={College of Control Science and Engineering, Zhejiang University},
	city={Hangzhou},
	postcode={310027}, 
	country={China}}



\cortext[cor1]{Corresponding author: {Yang\ Zhu}.}


\begin{abstract}
Extremum seeking (ES) is a real-time optimization
strategy, thus transmission delays in the feedback loop of ES
have big impact on its stability. How big delay that ES control
systems are able to withstand? This paper provides a
potential answer to this problem. We focus on gradient-based ES for nonlinear static maps subject to known constant delays plus a small time-varying delay uncertainty. {We also consider the measurement to be subject to a small disturbance.} Different from a majority of existing literature addressing quadratic maps with delays by predictor feedback, this paper deals with a wider class of non-quadratic maps without any predictor or observer for delay compensation. Dither signals in modulation and demodulation are carefully designed to handle constant delays and time-varying delay uncertainties. When the nonlinear map is unknown, we offer a rigorously analytical framework of ES convergence and delay-robustness. When some a prior knowledge of nonlinear maps is available, we are able to provide a quantitative estimation on upper bounds of time delay and dither periods to keep ES systems to remain stable. {A suitable choice of ES parameters guarantees practical stability for any large known constant delay.}
\end{abstract}




\maketitle

\section{Introduction}
Extremum seeking (ES) is a powerful real-time optimization method steering the output of an unknown map towards its local optimum \cite{SCHEINKER2024111481}. Since the milestone work \cite{KRSTIC2000595} in which ES convergence was proved in a seriously mathematical way, numerous new advances have emerged: non-local ES stability \cite{TAN2006889}, ES with unknown control direction or high relative degree \cite{GUAY2024,GUAY22024}, sampled-data ES \cite{KHONG20132720,HAZELEGER2022110415}, ES based on Lie-bracket \cite{DURR20131538,LABAR2022110041}, ES with measurement noise \cite{YangxuefeiJFI2025}, etc..

Time delay (arising from computation, measurement, transmission, etc.) is one of the most common phenomena in engineering practice, and when disregarded, they render controlled systems unstable \cite{FridmanBook2014,Yangbook2020,ZENG2014492,GE2019500,HU201932,SCHEINKER2024111481,ZHANG2024128925}. As a kind of real-time online optimal control algorithm, the impact of time lag on ES systems is particularly serious. In few cases, delays can be used to stabilize the ES system \cite{Suttner2024}. In the state-of-the-art, the tools to handle time-delay in ES could be divided into two streams:
\begin{itemize}
	\item The 1st method is predictor-based feedback in which delays are compensated via various sorts of predictors, e.g., predictors in the form of integral \cite{Rusiti2019,Damir2021,RUSITI202175,Oliveira7466811,tiagoacc2020,TSUBAKINO2023111044,Yilmaz2024}, predictors in the form of observer \cite{Georgeacc2020}, and sequential predictors \cite{MALISOFF2021109462}. The advantage of the predictor method is that delay length tolerated by ES control systems could be large, whereas the cost that you pay for this benefit is that such ES controllers are complicated.
	\item The 2nd method is predictor-free feedback where delays are not managed. By this technique, ES feedback systems depend upon its own robustness to delays to be stable \cite{yang2023extremum,YangxuefeiIJRNC2025,JBARA2025106256,ZHU2022tac}. Predictor-free ES controllers are much simpler than predictor-based ES controllers under the premise that delays are not large.
\end{itemize}

To determine which method that we should select, a key problem that we ask ourselves is: what is the maximum time delay magnitude that the ES system is able to withstand? If an ES controller is capable of bearing somewhat large delays, then there is no need to resort to predictors to make things complex. This paper bridges the gap about the lack of analysis of delay-robustness in ES control laws.

Unlike the publications  \cite{Oliveira7466811,tiagoacc2020,TSUBAKINO2023111044,Yilmaz2024,MALISOFF2021109462,yang2023extremum,YangxuefeiIJRNC2025,JBARA2025106256} and \cite{ZHU2022tac} taking care of quadratic maps of which the resulting averaged systems are linear time-delay systems, this paper manages with non-quadratic nonlinear static maps in the presence of uncertain time-varying delays. {The literature \cite{RUSITI202175} employed the predictor to compensate for a large constant delay of single-variable Newton-based ES of quadratic static maps under small delay mismatch. Taking advantage of a recently developed time-delay approach \cite{FRIDMAN2020109287,ZHU2022109965,gaofeng2023}, the present paper provides a constructive framework to investigate delay-robustness (without delay compensation) of classical multi-variable predictor-free ES of non-quadratic static maps in the presence of time-varying uncertain delays {and a small disturbance}.} {Recently, based on the time-delay approach, \cite{LI2026101782} presented two sampled-data ES schemes, classical unbiased ES and bounded unbiased ES, for which the admissible measurement delay must be sufficiently small compared with the algorithmic time scale, ensuring that the delay effect appears only as a higher-order perturbation in the averaging analysis. Through a novel design of dither signals, we are able to manage with ES convergence under large constant delays plus small time-varying delay uncertainties.} By transforming the original ES system into a kind of retarded differential equation (RDE) with distributed delays, we provide Lyapunov-based sufficient conditions in the form of linear matrix inequalities (LMIs) to ensure the practical stability of the ES time-delay system. Under a couple of non-restrictive assumptions on nonlinear maps, a solution of LMIs results in the quantitative upper bounds for the dither period and the time delay that the ES system is able to tolerate. {Appropriate ES parameters can be determined for large known delay and sufficiently small delay uncertainty.} {{In \cite{jianzhongAuto2025}, we offered a network-based two-dimensional (2D) ES control framework in the presence of sampling and uncertain time-varying communication delays. The present article focuses on the delay robustness of continuous-time $n$-dimensional (ND) ES {in the presence of large constant part of delay and small measurement bias}.}} A preliminary conference version of this paper had been presented in \cite{jianzhongcdc2024} where the case of constant delays with no time-varying delay uncertainty and disturbances was considered.

\textbf{Notation}: Throughout this paper, ${\mathbb R^n}$ represents $n$-dimensional Euclidean space with the Euclidean norm $|\cdot|$, ${\mathbb N}$ denotes the set of non-negative integers. The set of $n$th-order continuously differentiable functions is denoted by $C^n$. We use the superscript `${\rm T}$' to stand for matrix transposition. The 2-norm of a matrix ${A\in\mathbb R^{n\times n}}$ is denoted by $|A|=\sqrt{\lambda_{\rm max}(A^{\rm T}A)}$. The notation `$*$' denotes the symmetric elements of a symmetric matrix, and $X<0$ represents that $X$ is a negative definite matrix.

\section{Extremum Seeking of Scalar Case}
\begin{figure}[pos = h]
	\centering
	\includegraphics[scale=0.32]{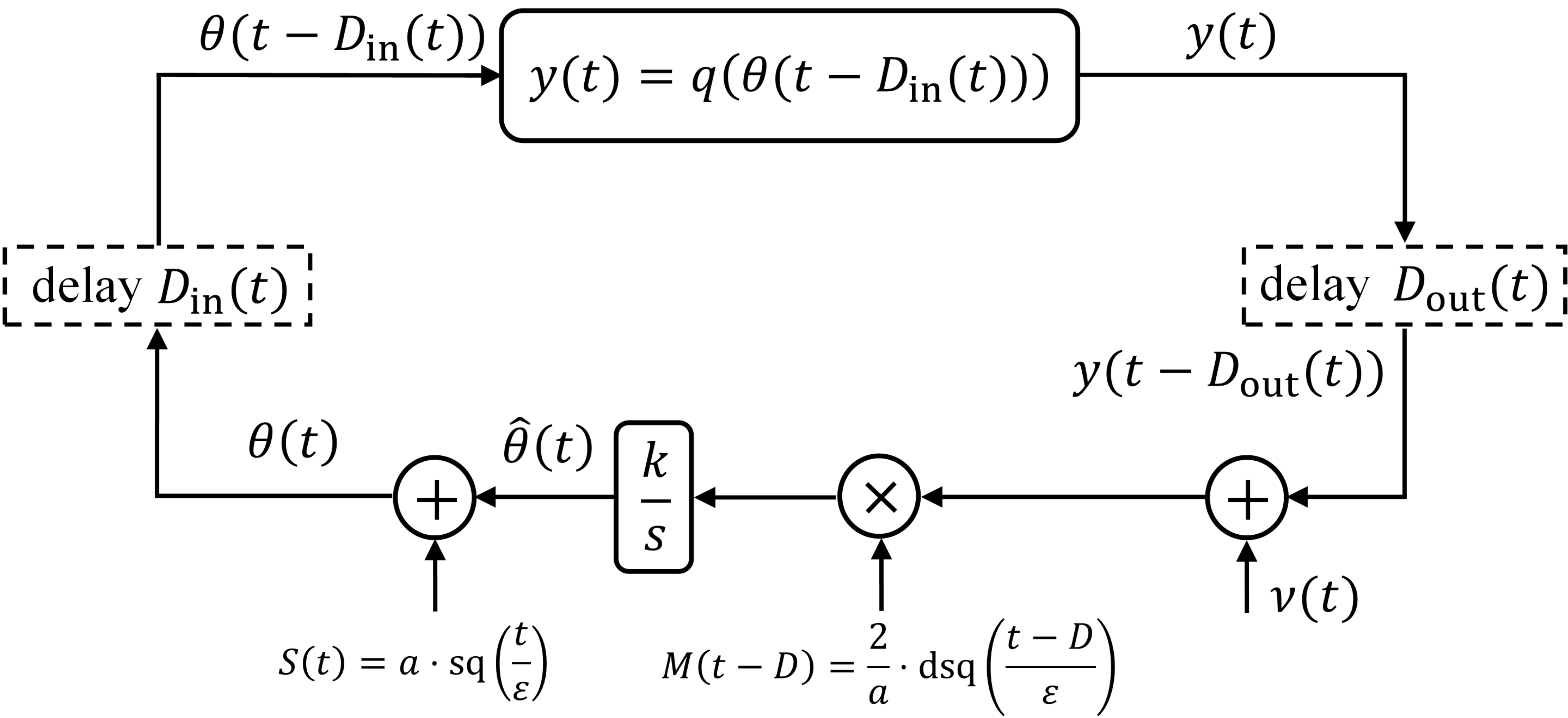}
	\DeclareGraphicsExtensions.
	\caption{{ES with {a disturbance and} time delays for static maps: scalar case.}}
	\label{figtvr_scalar}
\end{figure}
For conceptional clearness, we start with scalar maps in this section. Consider extremum seeking for nonlinear static maps
\begin{equation}\label{1tvr_map}
	\begin{array}{l}
		y(t)=q\big(\theta(t)\big),
	\end{array}
\end{equation}
where $y(t)\in \mathbb{R}$ is the measurable output, $\theta(t)\in \mathbb{R}$ is the scalar input, and $q(\cdot)$ is a nonlinear function. As depicted in Fig. \ref{figtvr_scalar}, {disturbance $\nu(t)\in \mathbb{R}$ characterizes the uncertainty associated with measurement bias, potentially induced by quantization in network-based control systems. The disturbance $\nu(t)$ is assumed to be bounded
\begin{equation}\label{bound_nu}
	\begin{array}{l}
		|\nu(t)|\leq \nu^*, \ \  \ t \geq 0,
	\end{array}
\end{equation}
with small, known constant $\nu^*\geq 0$.}
Moreover, the ES control scheme is subject to time-varying input delay $D_{\rm{in}}(t)$ and output delay $D_{\rm{out}}(t)$. We denote
\begin{equation}\label{tvrdelay}
	\begin{array}{l}
		h(t)=D_{\rm{in}}(t)+D_{\rm{out}}(t)=D+d(t),
	\end{array}
\end{equation}
where $D\geq0$ is a dominant constant delay assumed to be known, and $d(t)$ is a small time-varying delay uncertainty satisfying {the following relation }
\begin{equation}\label{tvrdelayub}
	\begin{array}{l}
		0\le d(t)\le\bar{d},
	\end{array}
\end{equation}
where $\bar{d}\geq0$ is a known upper bound. In this paper we do not seek to compensate for the time-varying delay $d(t)$ thus there is no need to know its specific form. From (\ref{tvrdelay})-(\ref{tvrdelayub}), we denote
\begin{equation}\label{tvrdelayadded}
	\begin{array}{l}
		\bar{h}\triangleq D+\bar{d}, \ \ \ \ 0\leq D\le h(t)\le\bar{h}.
	\end{array}
\end{equation}
From Fig. \ref{figtvr_scalar} and \eqref{tvrdelay}, it follows that
\begin{equation}\label{ey}
		y\left(t-D_{\rm{out}}(t)\right)=q\left(\theta\left(t-D_{\rm{out}}(t)-D_{\rm{in}}(t)\right)\right)
		 =q\left(\theta\left(t-h(t)\right)\right).
\end{equation}
The nonlinear map (\ref{1tvr_map}) satisfies the following assumptions.

\begin{assump} \cite{gaofeng2023}: There exist $\theta^*\in{\mathbb R}, \sigma>0$ and small $a>0$ such that $q(\theta)\in C^2(\theta^*-\sigma-a, \theta^*+\sigma+a)$, and the following hold:
\begin{align}
		&q'(\theta^*)=0,\ \ \ q''(\theta^*)=H<0, \label{1tvr_defineH} \\
		&q'(\theta^*+\Delta)\cdot \Delta \leq -\mu(\sigma)\cdot \Delta^2<0,\ \ \ \forall \ 0<|\Delta|<\sigma, \label{1tvr_boundsAssumption}
\end{align}
where $\mu(\sigma)>0$ is a known $\sigma$--dependent constant, which decreases monotonically in $\sigma$.  
\end{assump}

\begin{assump} \cite{gaofeng2023}: For any $|\Delta|<\sigma$ and $a$ defined in Assumption 1, given $\xi \in[-1, 1]$, we have
\begin{equation}\label{1tvr_AssumptionBound2}
	\begin{array}{l}
		|q(\theta^*+\Delta)|<q_0(\sigma),\\ |q'(\theta^*+\Delta)|<q_1(\sigma),\ \ \ |q''(\theta^*+\Delta)|<q_2(\sigma), \\ 
		|q'(\theta^*+\Delta)-q'(\theta^*)|<L_1|\Delta|,\\
		|q''(\theta^*+\Delta+a\xi)-q''(\theta^*)|<L_2|\Delta+a\xi|,
	\end{array}
\end{equation}
in which $q_0(\sigma), q_1(\sigma), q_2(\sigma), L_1$, and $ L_2$ are positive constants. 
\end{assump}

As clarified in \cite{gaofeng2023} (see Example 1 therein), Assumption 1 is not restrictive. It indicates the existence and uniqueness of the extremum $\theta=\theta^*$ within the neighborhood $U\left(\theta^*,\sigma+a\right)$. Here $a$ is the amplitude of the dither signal which will be employed in \eqref{1tvr_define_dithersignal}, whereas $\sigma$ defines a known search region which contains the unique extremum to be found. The knowledge of $\sigma$ is crucial for setting the initial condition to apply ES. This is similar to the local stability of a dynamical system in the sense that the system would never be stabilized if the initial state is out of the region of attraction. In the literature \cite{GUAY32024,Georgeacc2020}, the map (\ref{1tvr_map}) is usually unknown, the bounds on the map and its derivatives (\ref{1tvr_AssumptionBound2}) are unknown either, and they are used for analysis rather than design. In this article, when these bounds in (\ref{1tvr_AssumptionBound2}) are not known, we provide a qualitative analysis without approximation that is more accurate than the classical averaging-based analysis (see Remark 2). On the other hand, supposing that by some means a prior knowledge of the bounds (\ref{1tvr_AssumptionBound2}) are available (that is to say we are in the face of a ``grey box'' rather than a ``black box''), the developed time-delay approach allows a quantitative analysis by which we can calculate the upper bound on the key tuning parameter (the dither period given in (\ref{1tvr_define_dithersignal})), the ultimate bound of the ES estimation error (defined by (\ref{esterror})), and the biggest delay that the ES control system is able to withstand.

Under \eqref{ey}, the gradient ES controller is given below
\begin{equation}\label{1tvr_controller}
	\begin{aligned}
		{\theta}(t)&=\hat{\theta}(t)+S(t),\\
		\dot{\hat{\theta}}(t)&=k\cdot M(t-D)\cdot \left[y\left(t-D_{\rm{out}}(t)\right){{+\nu(t)}}\right]\\
		&=k\cdot M(t-D)\cdot \left[q\big(\theta\left(t-h(t)\right)\big){{+\nu(t)}}\right], \ \ \ \ t\ge \bar{h},
	\end{aligned}
\end{equation}
where $\hat{\theta}(t)$ is the real-time estimate of $\theta^*$ with the estimation error being
\begin{equation}\label{esterror}
	\tilde{\theta}(t)\triangleq\hat{\theta}(t)-\theta^*.
\end{equation}
The initial condition $\hat{\theta}(t)\in\left[\theta^*-\sigma_0,\theta^*+\sigma_0\right]$ and $\dot{\hat{\theta}}(t)\equiv0$ for $t\in\left[0,\bar{h}\right]$, where $\sigma_0<\sigma$ is a tuning parameter. It refers to the fact that, during the initial time interval $t\in\left[0,\bar{h}\right]$ whose length equals to the biggest transmission delay of the closed-loop, the ES controller does not work (the update law is turned off) so that the real-time estimate does not change \big(it is fixed as a constant in $\left[\theta^*-\sigma_0,\theta^*+\sigma_0\right]$\big) until the measurement of the output signal arrives at the controller side when $t=\bar{h}$. The sign of the adaptation gain $k$ is opposite to the sign of the Hessian $H$ in \eqref{1tvr_defineH}. The dither signals are of the forms
\begin{figure}[pos=H]
	\begin{center}
		\includegraphics[scale=0.3]{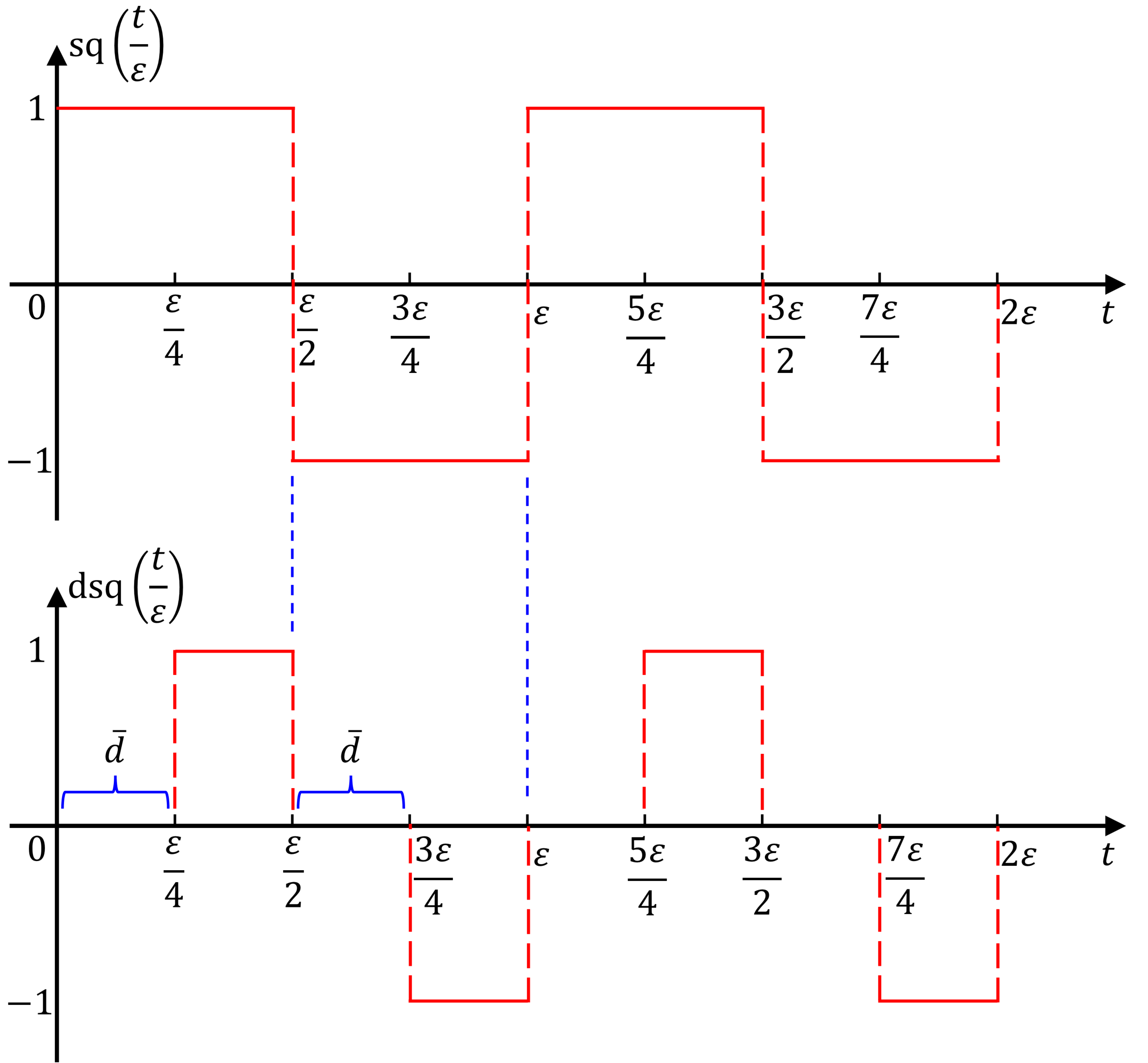}
		\DeclareGraphicsExtensions.
		\caption{{The square waves with the dither period satisfying $\varepsilon\geq4\bar{d}$}.}
		\label{1tvr_dither}
	\end{center}
\end{figure}
\begin{equation}
	\begin{array}{l}\label{1tvr_define_dithersignal}
		S(t)= {a\cdot\rm{sq}}\left(\frac{t}{\varepsilon}\right),\ \
		M(t-D)= \frac{2}{a}\cdot{\rm{dsq}}\left(\frac{t-D}{\varepsilon}\right),
	\end{array}
\end{equation}
with the dither amplitude $a>0$ and the dither period $\varepsilon>0$. 
As revealed in Fig. \ref{1tvr_dither}, the square waves ${\rm{sq}}\left(\frac{t}{\varepsilon}\right)$ and ${\rm{dsq}}\left(\frac{t}{\varepsilon}\right)$ are defined as
\begin{equation}\label{1tvr_define_dithersignal1}
	\begin{array}{l}
		{\rm{sq}}\left(\frac{t}{\varepsilon}\right)=\left\{\begin{array}{l}1,\ \ \ \   t\in \varepsilon\left[n, n+\frac{1}{2}\right), \vspace{0.4ex}\\
			-1,\ \   t\in \varepsilon\left[n+\frac{1}{2}, n+1\right),
		\end{array}  n\in {\mathbb N},
		\right.
	\end{array}
\end{equation}
\vspace{-0.6\baselineskip}
\begin{equation}\label{1tvr_define_dithersignal2}
	\begin{array}{l}
		{\rm{dsq}}\left(\frac{t}{\varepsilon}\right)=\left\{\begin{array}{l}1,\ \ \ \   t\in \varepsilon\left[n+\frac{1}{4}, n+\frac{1}{2}\right), \vspace{0.4ex}\\
			-1,\ \  t\in \varepsilon\left[n+\frac{3}{4}, n+1\right),\\
			0,\ \ \ \ \ \rm others,
		\end{array}  n\in {\mathbb N}.
		\right.
	\end{array}
\end{equation}
The phase-shifted dither $M(t-D)$ is introduced to address the constant delay $D$, and the signal ${\rm{dsq}}\left(\frac{t}{\varepsilon}\right)$ is tailored to address the unknown time-varying delay $d(t)$ in \eqref{tvrdelay}-(\ref{tvrdelayub}). We can select $\varepsilon>2\bar{d}$ to ensure that the 1st equality of \eqref{1tvr_integerdsq} given below is greater than zero (see the last paragraph of Remark 1). On the other hand, when the dither period $\varepsilon$ is given, the largest delay uncertainty allowed by the ES system satisfies $0\leq d(t)\le\bar{d}<\frac{\varepsilon}{2}$. For computational simplicity, in Section II of this paper, we let $\varepsilon\geq4\bar{d}$. Note that, the multiplicative demodulation signal $M(t-D)= \frac{2}{a}\cdot{\rm{dsq}}\left(\frac{t-D}{\varepsilon}\right)$ in \eqref{1tvr_define_dithersignal} invented to handle the time-varying uncertain delays is more sophisticated than the dither $M(t-D)= \frac{2}{a}\cdot{\rm{sq}}\big(\frac{t-D}{\varepsilon}\big)$ in the case of constant delays \cite{ZHU2022tac,jianzhongcdc2024}.

Under (\ref{1tvr_controller})-(\ref{esterror}), consider the Taylor expansion
\begin{equation}\label{1tvr_ytd}
	\begin{array}{l}
		q\big(\theta\left(t-h(t)\right)\big) =q\left(\theta^{*}+\tilde{\theta}\left(t-h(t)\right)+a{\rm{sq}}\left(\frac{t-h(t)}{\varepsilon}\right)\right)\\ \ \ \ \ \ \ \ \ \ \ \ \ \ \ \ \ \ \ \  \ \ \ \ \ 
		=q\left(\theta^{*}+\tilde{\theta}(t-h(t))\right)\\ \ \ \ \ \ \ \ \ \ \ \ \ \ \ \ \ \ \ \  \ \ \ \ \ \ \ \ \ +q'\left(\theta^{*}+\tilde{\theta}(t-h(t))\right)a{\rm{sq}}\left(\frac{t-h(t)}{\varepsilon}\right) +\frac{1}{2}H(t)a^2{\rm{sq^2}}\left(\frac{t-h(t)}{\varepsilon}\right),
	\end{array}
\end{equation}
where $H(t)\triangleq q''\left(\theta^{*}+\tilde{\theta}(t-h(t))+\zeta a{\rm{sq}}\left(\frac{t-h(t)}{\varepsilon}\right)\right)$, $\zeta\in(0,1)$. Then, substituting \eqref{1tvr_ytd} into \eqref{1tvr_controller}, and taking the time-derivative of \eqref{esterror} along \eqref{1tvr_controller}, the dynamics of the estimation error is governed by
\begin{equation}\label{1tvr_errorSystem1}
	\begin{array}{l}
		\dot{\tilde{\theta}}(t)=\dot{\hat{\theta}}(t)
		=k\cdot\frac{2}{a}\cdot{\rm{dsq}}\left(\frac{t-D}{\varepsilon}\right)\cdot \left[q\big(\theta\left(t-h(t)\right)\big){{+\nu(t)}}\right]\\
		=\frac{2k}{a}{\rm{dsq}}\left(\frac{t-D}{\varepsilon}\right)\Big[q\left(\theta^{*}+\tilde{\theta}\left(t-h(t)\right)\right) +q'\left(\theta^{*}+\tilde{\theta}\left(t-h(t)\right)\right)a{\rm{sq}}\left(\frac{t-h(t)}{\varepsilon}\right)+\frac{1}{2}H(t)a^2 {+\nu(t)}\Big]\\ =\frac{2k}{a}q\left(\theta^{*}+\tilde{\theta}\left(t-h(t)\right)\right){\rm{dsq}}\left(\frac{t-D}{\varepsilon}\right)\\
		\ \ \ \  +2kq'\left(\theta^{*}+\tilde{\theta}\left(t-h(t)\right)\right){\rm{dsq}}\left(\frac{t-D}{\varepsilon}\right){\rm{sq}}\left(\frac{t-h(t)}{\varepsilon}\right)
		\\ \ \ \ \ +kaH(t){\rm{dsq}}\left(\frac{t-D}{\varepsilon}\right) {+\frac{2k}{a}\nu(t){\rm{dsq}}\left(\frac{t-D}{\varepsilon}\right)},
	\end{array}
\end{equation}
where we used ${\rm{sq^2}}\left(\frac{t-h(t)}{\varepsilon}\right)=1$. We rewrite \eqref{1tvr_errorSystem1} as follows:
\begin{equation}\label{1tvr_errorSystemFR}
	\dot{\tilde{\theta}}(t)=k\mathcal{F}(t)+ka\mathcal{R}(t) {+k\mathcal{W}(t)},\ \ t\ge \bar{h},
\end{equation}
where
\begin{equation}\label{1tvr_DefineFR}
	\begin{array}{l}
		\mathcal{F}(t)\triangleq
		\frac{2}{a}q\left(\theta^{*}+\tilde{\theta}\left(t-h(t)\right)\right){\rm{dsq}}\left(\frac{t-D}{\varepsilon}\right) +2q'\left(\theta^{*}+\tilde{\theta}\left(t-h(t)\right)\right){\rm{dsq}}\left(\frac{t-D}{\varepsilon}\right){\rm{sq}}\left(\frac{t-h(t)}{\varepsilon}\right)\\ \ \ \ \ \ \ \ \ \ \ \ \ +aH{\rm{dsq}}\left(\frac{t-D}{\varepsilon}\right)
		,\\
		\mathcal{R}(t)\triangleq \left(H(t)-H\right){\rm{dsq}}\left(\frac{t-D}{\varepsilon}\right),\\
		{\mathcal{W}(t)\triangleq \frac{2}{a}\nu(t){\rm{dsq}}\left(\frac{t-D}{\varepsilon}\right)}.
	\end{array}
\end{equation}
In the classical averaging-like analysis, selecting the dither period $\varepsilon$ to be small enough, the terms $q\left(\theta^{*}+\tilde{\theta}\left(t-h(t)\right)\right)$ and $q'\left(\theta^{*}+\tilde{\theta}\left(t-h(t)\right)\right)$ are slowly time-varying comparative to the dither and are approximated as ``freezing'' constants. Besides, the term $ka\mathcal{R}(t)$ on the right-hand side of \eqref{1tvr_errorSystemFR} is neglected if the dither amplitude $a$ is small. The delayed dither signals are shown in Fig. \ref{1tvr_dither_delayed}. Therefore, we formally obtain an averaged-like system of \eqref{1tvr_errorSystemFR} such that
\begin{equation}\label{1averagedsystem}
	\begin{array}{l}
		\dot{\tilde{\theta}}_{av}(t)=\frac{2k}{a\varepsilon}\int_{t-\varepsilon}^t {\rm{dsq}}\left(\frac{\tau-D}{\varepsilon}\right)d\tau\cdot q\left(\theta^{*}+\tilde{\theta}_{av}\left(t-h(t)\right)\right)\\
		\ \ \  \ \ \ \ \ \  \ \ \ \ +\frac{2k}{\varepsilon}\int_{t-\varepsilon}^t{\rm{dsq}}\left(\frac{\tau-D}{\varepsilon}\right){\rm{sq}}\left(\frac{\tau-h(\tau)}{\varepsilon}\right)d\tau \cdot q'\left(\theta^{*}+\tilde{\theta}_{av}\left(t-h(t)\right)\right)\\
		\ \ \ \ \ \ \ \ \ \ \ \ \ +\frac{ak}{\varepsilon}\int_{t-\varepsilon}^t {\rm{dsq}}\left(\frac{\tau-D}{\varepsilon}\right)d\tau\cdot H {+k\mathcal{W}(t)}\\ \ \ \ \ \ \ \ \ \ \   
		=kq'\left(\theta^{*}+\tilde{\theta}_{av}\left(t-h(t)\right)\right) {+k\mathcal{W}(t)},
	\end{array}
\end{equation}
where we utilized the calculation
\begin{equation}\label{1tvr_integerdsq}
	\begin{array}{l}
		\frac{1}{\varepsilon}\int_{t-\varepsilon}^{t}{\rm{dsq}}
		\left(\frac{\tau-D}{\varepsilon}\right){\rm{sq}}\left(\frac{\tau-h(\tau)}{\varepsilon}\right){\rm{d}}\tau=\frac{1}{2}, \\ 
		\frac{1}{\varepsilon}\int_{t-\varepsilon}^{t}{\rm{dsq}}
		\left(\frac{\tau-D}{\varepsilon}\right){\rm{d}}\tau=0.
	\end{array}
\end{equation}
\begin{figure}[pos=H]
	\centering
		\includegraphics[scale=0.30]{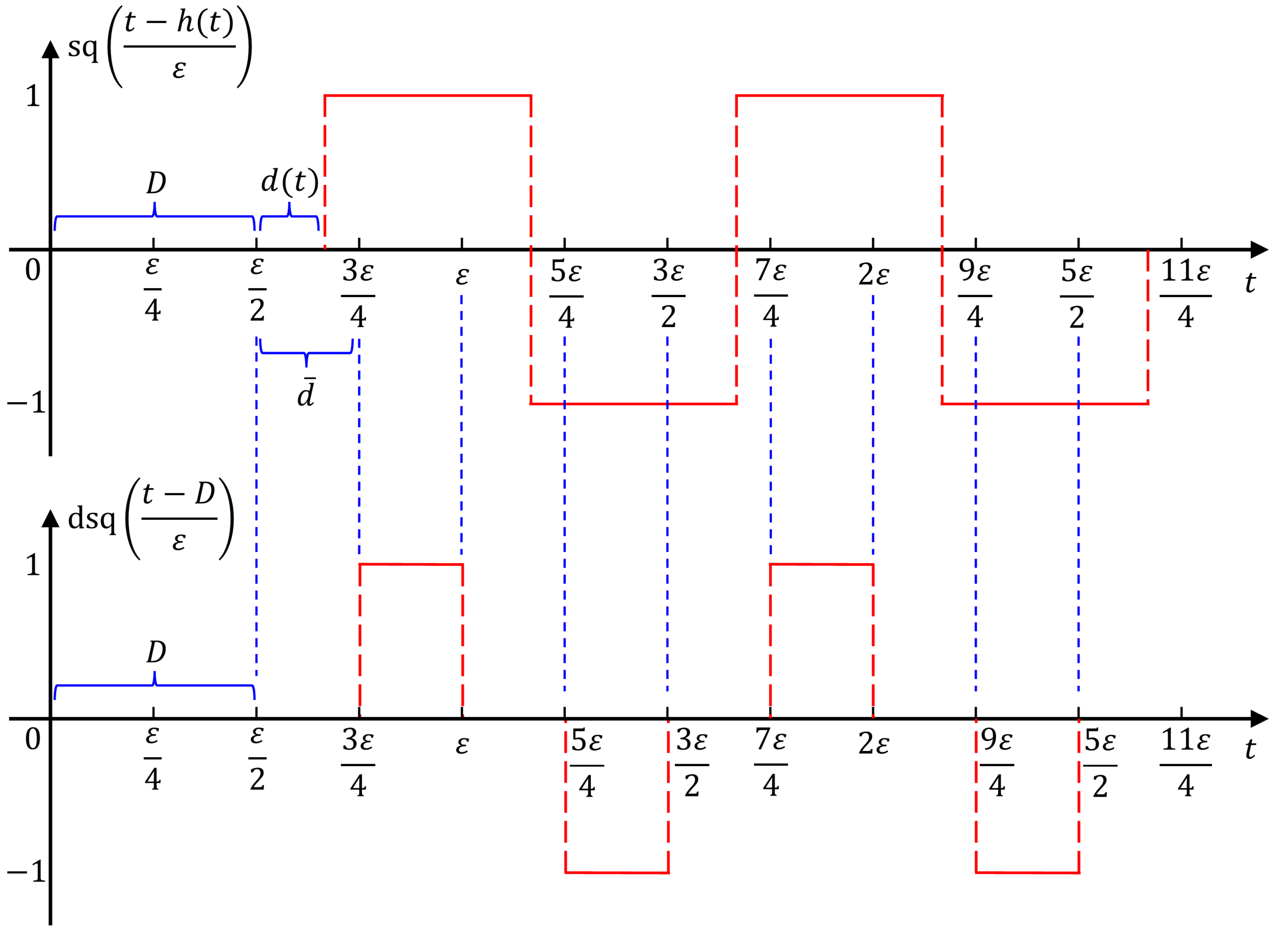}
		\caption{{The dither signals under $D=\frac{\varepsilon}{2}$ and $\bar{d}\leq\frac{\varepsilon}{4}$}.}
		\label{1tvr_dither_delayed}
\end{figure}
\begin{rmk}
Supposing ${\mathcal{W}(t)\equiv 0}$ and $h(t)\equiv0$ in \eqref{1averagedsystem} (the delay-free case), under \eqref{1tvr_boundsAssumption} in Assumption 1 and $k>0$, if $\tilde{\theta}_{av}(\cdot)>0$, then the $q'\big(\theta^*+\tilde{\theta}_{av}(\cdot)\big)<0$ so that the derivative $\dot{\tilde{\theta}}_{av}(t)<0$, which indicates $\tilde{\theta}_{av}(\cdot)$ will decrease to zero, and vice versa. Consequently, the delay-free averaged system is stable. When $h(t)>0$ in \eqref{1averagedsystem}, the research community are concerned with the robustness to the
length of delay such that the nonlinear time-delay system \eqref{1averagedsystem} still can be kept stable if the delay is not large. We first consider the case of constant delays such that $h(t)\equiv D$ and we omit the subscript ``\textit{av}'' and $\theta^*$ in \eqref{1averagedsystem} for notational simplicity. If the nonlinear system
\begin{equation}\label{avdk}
	\begin{array}{l}
		\dot{\tilde{\theta}}(t)=kq'\left(\tilde{\theta}\left(t-D\right)\right)
	\end{array}
\end{equation}
is stable for a fixed pair $(k,D)$ with $k>0$ and $D\in\left[0,\bar{D}\right]$, then theoretically the system \eqref{avdk} can be kept stable for any large delay $D$, but with almost zero decay rate $k$ (see \cite[Example 2.3]{FridmanBook2014}). To be specific, denoting $t=\omega \tau$ where $\omega>0$ is a scaling constant, and $\tilde{\theta}_s\left(\tau\right)=\tilde{\theta}\left(\omega\tau\right)=\tilde{\theta}(t)$, we have $\tilde{\theta}(t-D)=\tilde{\theta}\left(\omega\tau-D\right)=\tilde{\theta}\left(\omega\left(\tau-\frac{D}{\omega}\right)\right)=\tilde{\theta}_s\left(\tau-\frac{D}{\omega}\right)$ and
\begin{equation}\label{avdk2}
	\begin{array}{l}
		\frac{d}{d\tau}\tilde{\theta}_s\left(\tau\right)=\omega\frac{d}{dt}\tilde{\theta}\left(t\right)=\omega kq'\left(\tilde{\theta}(t-D)\right)\\
		=\omega kq'\left(\tilde{\theta}_s\left(\tau-\frac{D}{\omega}\right)\right)=k_sq'\left(\tilde{\theta}_s\left(\tau-D_s\right)\right),
	\end{array}
\end{equation}
where $k_s=\omega k$, $D_s=\frac{D}{\omega}$. The system \eqref{avdk2} recovers the form of \eqref{avdk} and is stable as long as $0\le D_s=\frac{D}{\omega}\le \bar{D}$, i.e., $0\le D=\omega D_s\le \omega\bar{D}$. In other words, the delay $D$ is magnified $\omega$ times whereas the control gain $k$ is shrunk $\omega$ times. The scaling coefficient $\omega$ could be arbitrarily huge in principle. For quadratic maps in which the averaged system is linear \cite{yang2023extremum}, the designers can decrease the decay rate $k$ and manage with any constant delay $D$. Here we have the similar findings for non-quadratic maps that we also manage with any constant delay but with a smaller decay rate. In practice, we should ensure the real-time convergence rate so that the control gain cannot be arbitrarily small. Thus, provided a given $k$, we are interested in the upper bound of delay $D$ that the ES system is able to tolerate.

Besides, it is worth mentioning that the dither signals \eqref{1tvr_define_dithersignal1}-\eqref{1tvr_define_dithersignal2} are meticulously designed to guarantee that the average of the dominant stabilization term (corresponding to the 2nd  line of \eqref{1averagedsystem}) is not zero, whereas the average of the perturbation terms (which refers to the 1st and 3rd lines of \eqref{1averagedsystem}) is zero. Under the present dithers shown in Figs. \ref{1tvr_dither}-\ref{1tvr_dither_delayed}, the convergence rate of \eqref{1averagedsystem} is $k$, and the largest delay uncertainty allowed is $\bar{d}=\frac{\varepsilon}{4}$. When the upper bound of the time-varying delay uncertainty is larger (e.g., $\bar{d}=\frac{\varepsilon}{3}$), the convergence rate is $\frac{2}{3}k$ which is smaller. A trade-off is there. 
\end{rmk}

Next, we apply the time-delay approach {(refer to, \cite{FRIDMAN2020109287,ZHU2022109965})} to the ES system \eqref{1tvr_errorSystemFR}. To be specific, we integrate $\mathcal{F}(t)$ over the interval $[t-\varepsilon, t]$ for $t\geq \varepsilon+\bar{h}$ such that
\begin{equation}\label{1tvr_defineF}
	\begin{array}{l}
		\frac{1}{\varepsilon}\int_{t-\varepsilon}^{t}\mathcal{F}(\tau){\rm{d}}\tau=\frac{2}{a\varepsilon}
		\int_{t-\varepsilon}^{t}q\big(\theta^{*}+\tilde{\theta}(\tau-h(\tau))\big){\rm{dsq}}\left(\frac{\tau-D}{\varepsilon}\right){\rm{d}}\tau\\ \ \ \ \ \ \ \ \ \ \ \ \ \ \ \ \ \ \ \ \ \ \ \ \ \ \ \ +\frac{2}{\varepsilon}\int_{t-\varepsilon}^{t}
		q'\big(\theta^{*}+\tilde{\theta}(\tau-h(\tau))\big){\rm{dsq}}\left(\frac{\tau-D}{\varepsilon}\right){\rm{sq}}\left(\frac{\tau-h(\tau)}
		{\varepsilon}\right){\rm{d}}\tau,
	\end{array}
\end{equation}
in which the condition $\frac{a}{\varepsilon}\int_{t-\varepsilon}^t {\rm{dsq}}\left(\frac{\tau-D}{\varepsilon}\right)d\tau\cdot H=0$ is employed.
To cope with the 1st term on the right-hand side of (\ref{1tvr_defineF}), using $\int_{t-\varepsilon}^{t}{\rm{dsq}}\left(\frac{\tau-D}{\varepsilon}\right){\rm{d}}\tau \cdot q\big(\theta^{*}+\tilde{\theta}(t-D)\big)=0$, it follows that
\begin{equation}\label{1tvr_integerF2}
	\begin{array}{l}
		\frac{2}{a\varepsilon}\int_{t-\varepsilon}^{t}q\big(\theta^{*}+\tilde{\theta}(\tau-h(\tau))\big){\rm{dsq}}\left(\frac{\tau-D}{\varepsilon}\right){\rm{d}}\tau\\
		=-\frac{2}{a\varepsilon}
		\int_{t-\varepsilon}^{t}{\rm{dsq}}\left(\frac{\tau-D}{\varepsilon}\right)\left[q\big(\theta^{*}+\tilde{\theta}(t-D)\big)-q\big(\theta^{*}+
		\tilde{\theta}(\tau-h(\tau))\big)\right]{\rm{d}}\tau\\
		=-\frac{2}{a\varepsilon}
		\int_{t-\varepsilon}^{t}{\rm{dsq}}\left(\frac{\tau-D}{\varepsilon}\right)\int_{\tau-h(\tau)}^{t-D}q'\big(\theta^{*}+
		\tilde{\theta}(s)\big)\dot{\tilde{\theta}}(s){\rm{d}}s{\rm{d}}\tau.
	\end{array}
\end{equation}
To deal with the 2nd term on the right-hand side of (\ref{1tvr_defineF}), taking into account the 1st formula of (\ref{1tvr_integerdsq}), we have
\begin{equation}\label{1tvr_integerF1}
	\begin{array}{l}
		\frac{2}{\varepsilon}\int_{t-\varepsilon}^{t}
		q'\big(\theta^{*}+\tilde{\theta}(\tau-h(\tau))\big){\rm{dsq}}\left(\frac{\tau-D}{\varepsilon}\right){\rm{sq}}\left(\frac{\tau-h(\tau)}
		{\varepsilon}\right){\rm{d}}\tau\\
		=\frac{2}{\varepsilon}\int_{t-\varepsilon}^{t}
		{\rm{dsq}}\left(\frac{\tau-D}{\varepsilon}\right){\rm{sq}}\left(\frac{\tau-h(\tau)}
		{\varepsilon}\right)\left[q'\big(\theta^{*}+\tilde{\theta}(\tau-h(\tau))\big)\right.\\
		\ \ \left.- q'\big(\theta^{*}+\tilde{\theta}(t-D)\big)+q'\big(\theta^{*}+\tilde{\theta}(t-D)\big)\right]{\rm{d}}\tau\\
		=\frac{2}{\varepsilon}\int_{t-\varepsilon}^{t}
		{\rm{dsq}}\left(\frac{\tau-D}{\varepsilon}\right){\rm{sq}}\left(\frac{\tau-h(\tau)}
		{\varepsilon}\right){\rm{d}}\tau \cdot q'\big(\theta^{*}+\tilde{\theta}(t-D)\big)\\
		\ \  \ -\frac{2}{\varepsilon}\int_{t-\varepsilon}^{t}
		{\rm{dsq}}\left(\frac{\tau-D}{\varepsilon}\right){\rm{sq}}\left(\frac{\tau-h(\tau)}
		{\varepsilon}\right)\left[q'\big(\theta^{*}+\tilde{\theta}(t-D)\big)-q'\big(\theta^{*}+\tilde{\theta}(\tau-h(\tau))\big)\right]{\rm{d}}\tau\\
		=q'\big(\theta^{*}+\tilde{\theta}(t-D)\big)-\frac{2}{\varepsilon}\int_{t-\varepsilon}^{t}
		{\rm{dsq}}\left(\frac{\tau-D}{\varepsilon}\right){\rm{sq}}\left(\frac{\tau-h(\tau)}
		{\varepsilon}\right) \int_{\tau-h(\tau)}^{t-D}q''\big(\theta^{*}+
		\tilde{\theta}(s)\big)\dot{\tilde{\theta}}(s){\rm{d}}s{\rm{d}}\tau.
	\end{array}
\end{equation}
To transform the ES system \eqref{1tvr_errorSystemFR} into a kind of time-delay system, we define
\begin{equation}\label{1tvr_DefineG}
	\begin{array}{l}
		G(t)=\frac{1}{\varepsilon}
		\int_{t-\varepsilon}^{t}(\tau-t+\varepsilon)\mathcal{F}(\tau){\rm{d}}\tau,
	\end{array}
\end{equation}
and the following relation holds
\begin{equation}\label{1tvr_systemFF}
	\begin{array}{l}
		\frac{\rm{d}}{{\rm{d}}t}\Big[{\tilde{\theta}}(t)-kG(t)\Big]=\dot{\tilde{\theta}}(t)-k\mathcal{F}(t)+\frac{k}{\varepsilon}
		\int_{t-\varepsilon}^{t}\mathcal{F}(\tau){\rm{d}}\tau.
	\end{array}
\end{equation}

Substituting (\ref{1tvr_integerF2})-(\ref{1tvr_integerF1}) into \eqref{1tvr_defineF}, and further substituting (\ref{1tvr_errorSystemFR}) and \eqref{1tvr_defineF} into (\ref{1tvr_systemFF}), the closed-loop error system is given below
\begin{equation}\label{1tvr_errorSystemThetaYR}
	\begin{array}{l}
		\frac{\rm{d}}{{\rm{d}}t}\Big[{\tilde{\theta}}(t)-kG(t)\Big]=kq'\big(\theta^{*}+\tilde{\theta}(t-D)\big)-\frac{2k}{a}Y_1(t)
		 -2kY_2(t)+ka\mathcal{R}(t){+k\mathcal{W}(t)},\ \ \  t\geq \varepsilon+\bar{h},
	\end{array}
\end{equation}
where
\begin{equation}\label{1tvr_DefineY}
	\begin{array}{l}
		Y_1(t)\triangleq\frac{1}{\varepsilon}
		\int_{t-\varepsilon}^{t}\int_{\tau-h(\tau)}^{t-D}{\rm{dsq}}\left(\frac{\tau-D}{\varepsilon}\right)q'\big(\theta^{*}+
		\tilde{\theta}(s)\big)\dot{\tilde{\theta}}(s){\rm{d}}s{\rm{d}}\tau,\\
		Y_2(t)\triangleq\frac{1}{\varepsilon}\int_{t-\varepsilon}^{t}
		\int_{\tau-h(\tau)}^{t-D}{\rm{dsq}}\left(\frac{\tau-D}{\varepsilon}\right){\rm{sq}}\left(\frac{\tau-h(\tau)}
		{\varepsilon}\right) q''\big(\theta^{*}+
		\tilde{\theta}(s)\big)\dot{\tilde{\theta}}(s){\rm{d}}s{\rm{d}}\tau,
	\end{array}
\end{equation}
with $\dot{\tilde{\theta}}(t)$ defined by (\ref{1tvr_errorSystemFR}). Notice that, under Assumption 2, if ${\tilde{\theta}}(t)$ can be kept bounded and the bound is of order $O\left(1\right)$, {then $\dot{\tilde{\theta}}(t)$ is of order $O\left(\frac{1}{a}\right)$. Hence, in the system \eqref{1tvr_errorSystemThetaYR}, the disturbance term $G(t)$ is of order $O\left(\frac{\varepsilon}{a}\right)$, the term $\frac{2}{a}Y_1(t)$ is of order $O\left(\frac{\varepsilon+\bar{d}}{a^2}\right)$, and $Y_2(t)$ is of order $O\left(\frac{\varepsilon+\bar{d}}{a}\right)$. In addition, the Taylor remainder term $a\mathcal{R}(t)$ is of order $O(a)$. Selecting $a={(\varepsilon+\bar{d})}^{\frac{1}{3}}$, all the disturbances $G(t)$, $\frac{2}{a}Y_1(t)$, $Y_2(t)$, $a\mathcal{R}(t)$, {and $\mathcal{W}(t)$} decay to zero when $\varepsilon$, $\bar{d}$, {and $\nu^*$} are sufficiently small.}

\begin{rmk} Comparing the averaging-based analysis (namely, the transformation \eqref{1tvr_errorSystemFR}$\rightarrow$\eqref{1averagedsystem}) with the time-delay approach (i.e., the transformation \eqref{1tvr_errorSystemFR}$\rightarrow$\eqref{1tvr_errorSystemThetaYR}), it is evident that the latter does not employ any approximation. Unlike the averaging method in \eqref{1averagedsystem}-\eqref{1tvr_integerdsq}, the terms $q\left(\theta^{*}+\tilde{\theta}\left(t-h(t)\right)\right)$ and $q'\left(\theta^{*}+\tilde{\theta}\left(t-h(t)\right)\right)$ are not treated as ``freezing'' constants and not put outside of the integral in \eqref{1tvr_defineF}, and the term $ka\mathcal{R}(t)$ is not ignored. {The dominant part of \eqref{1tvr_errorSystemThetaYR}, i.e., $\frac{\rm{d}}{{\rm{d}}t}{\tilde{\theta}}(t)=kq'\big(\theta^{*}+\tilde{\theta}(t-D)\big)$, is consistent with the averaged system \eqref{1averagedsystem} (when $d(t)=0$).} For the original ES system \eqref{1tvr_errorSystemFR}, the time-delay plant \eqref{1tvr_errorSystemThetaYR} is a more precise representation than the averaged system \eqref{1averagedsystem} in the sense that the disturbances $G(t), Y_1(t), Y_2(t)$, and $a\mathcal{R}(t)$ are expressed in an explicit way. As a result, the stability of the original ES system \eqref{1tvr_errorSystemFR} can be concluded from the stability of the distributed time-delay system \eqref{1tvr_errorSystemThetaYR}. \end{rmk}

For analysis simplicity, we further employ the following change of variable \cite{YangxuefeiCDC2023},
\begin{equation*}\label{1tvr_denoteZ}
	z(t)=\tilde{\theta}(t)-kG(t).
\end{equation*}
Then the time-delay system of neutral type (\ref{1tvr_errorSystemThetaYR}) is transformed into the ordinary differential equation (ODE) with disturbance terms as follows:
\begin{equation}\label{1tvr_errorSystem_Phi_ZYR}
	\begin{array}{l}
		\dot z(t)=kq'\big(\theta^{*}+z(t-D)+kG(t-D)\big)
		-\frac{2k}{a}Y_1(t) -2kY_2(t)+ka\mathcal{R}(t){+k\mathcal{W}(t)}
		\\ \ \ \ \ \ \ \   =kq'\big(\theta^{*}+z(t-D)\big)+k\Phi(t)
		-\frac{2k}{a}Y_1(t) -2kY_2(t)+ka\mathcal{R}(t){+k\mathcal{W}(t)},\ \ \  t\geq \varepsilon+\bar{h},
	\end{array}
\end{equation}
where
\begin{equation}\label{1tvr_define_Phi}
	\begin{array}{l}
		\Phi(t)=q'\big(\theta^{*}+z(t-D)+kG(t-D)\big) -q'\big(\theta^{*}+z(t-D)\big).
	\end{array}
\end{equation}
Keeping in mind $\sigma$ defined in Assumption 1, we postulate the overall bound of the estimation error satisfies
\begin{equation}\label{1tvr_overallBound}
	\begin{array}{l}
		\Big|\tilde{\theta}(t)\Big|<\sigma,\ \  t\geq 0.
	\end{array}
\end{equation}
Then, given the definitions
of (\ref{1tvr_controller})-(\ref{1tvr_define_dithersignal}), we have the relation: $\theta^*-\sigma-a<\theta(t)=\theta^*+\tilde{\theta}(t)+a{\rm{sq}}\left(\frac{t}{\varepsilon}\right)<\theta^*+\sigma+a$, which is consistent with the domain specified in Assumption 1. With the initial condition underneath \eqref{esterror}, the overall bound of real-time estimate error \eqref{1tvr_overallBound} will always be guaranteed by the LMI conditions (\ref{1tvr_lmi1}) in Theorem 1. It is worth mentioning that the overall bound (\ref{1tvr_overallBound}) is distinct from the ultimate bound (\ref{ubscalar}) in Theorem 1 which is adjustable via the tuning parameters $\varepsilon$ and $a$ (see the illustrative Fig. 4 in \cite{gaofeng2023}).

Under the bounds (\ref{1tvr_AssumptionBound2}) in Assumption 2, and the overall bound (\ref{1tvr_overallBound}), the upper bounds on $\mathcal{F}(t), \mathcal{R}(t)$, and $\dot{\tilde{\theta}}(t)$ are obtained from (\ref{1tvr_errorSystemFR})-(\ref{1tvr_DefineFR}) such that
\begin{equation}\label{1tvr_boundFR}
	\begin{array}{l}
		\left|\mathcal{F}(t)\right|\leq \frac{2}{a}\left|q\big(\theta^{*}+\tilde{\theta}(t-h(t))\big)\right|+2\left|q'\big(\theta^{*}+\tilde{\theta}(t-h(t))\big)\right| +a|H|
		\\
		\ \ \ \ \ \  \ \ \ \  <\frac{2}{a}q_0(\sigma)+2q_1(\sigma)+aq_2(\sigma)\triangleq \Delta_\mathcal{F},\\
		|\mathcal{R}(t)|\leq  \left|q''\left(\theta^{*}+\tilde{\theta}(t-h(t))+\zeta a{\rm{sq}}\left(\frac{t-h(t)}{\varepsilon}\right)\right)-q''\left(\theta^{*}\right)\right|\\
		\ \ \ \ \ \  \ \ \ \  <L_2\left|\tilde{\theta}(t-h(t))+\zeta a{\rm{sq}}\big(\frac{t-h(t)}{\varepsilon}\big)\right|  <L_2(\sigma+a)\triangleq \Delta_\mathcal{R},\\
		{|\mathcal{W}(t)|=\left|\frac{2}{a}\nu(t){\rm{dsq}}\big(\frac{t-D}{\varepsilon}\big)\right|\leq \frac{2}{a}\nu^*},\\
		\left|\dot{\tilde{\theta}}(t)\right|\leq |k\mathcal{F}(t)|+|ka\mathcal{R}(t)| {+|k\mathcal{W}(t)|}<k\big(\Delta_\mathcal{F}+a\Delta_\mathcal{R} {+\frac{2}{a}\nu^*}\big)\triangleq \Delta_\theta.
	\end{array}
\end{equation}
Accordingly, based on (\ref{1tvr_DefineG}), (\ref{1tvr_DefineY}), and \eqref{1tvr_errorSystem_Phi_ZYR}, it can be derived that
\begin{equation}\label{1tvr_boundsGY}
	\begin{array}{l}
		|G(t)|= \frac{1}{\varepsilon}\left|\int_{t-\varepsilon}^{t}(\tau-t+\varepsilon)\mathcal{F}(\tau){\rm{d}}\tau\right|
		\\ \ \ \ \ \ \  \ \ \ \  \leq \frac{1}{\varepsilon}\int_{t-\varepsilon}^{t}(\tau-t+\varepsilon){\rm{d}}\tau\cdot|\mathcal{F}(\tau)|<\frac{\Delta_\mathcal{F}}{2}\varepsilon,\\
		|Y_1(t)|=\frac{1}{\varepsilon}
		\left|\int_{t-\varepsilon}^{t}\int_{\tau-h(\tau)}^{t-D}{\rm{dsq}}\left(\frac{\tau-D}{\varepsilon}\right)q'\big(\theta^{*}+
		\tilde{\theta}(s)\big) \dot{\tilde{\theta}}(s){\rm{d}}s{\rm{d}}\tau\right|\\ \ \ \ \ \ \  \ \ \ \ \
		\leq \frac{1}{\varepsilon}\int_{t-\varepsilon}^{t}
		\int_{\tau-h(\tau)}^{t-D}{\rm{d}}s{\rm{d}}\tau\cdot \left|q'\big(\theta^{*}+
		\tilde{\theta}(s)\big)\right| \cdot \left|\dot{\tilde{\theta}}(s)\right|\\  \ \ \ \ \ \  \ \ \ \ \
		<\frac{1}{2}{q_1(\sigma)\Delta_\theta}(\varepsilon+2\bar{d}),\\
		|Y_2(t)|=\frac{1}{\varepsilon}\left|\int_{t-\varepsilon}^{t}
		\int_{\tau-h(\tau)}^{t-D}{\rm{dsq}}\left(\frac{\tau-D}{\varepsilon}\right){\rm{sq}}\big(\frac{\tau-h(\tau)}
		{\varepsilon}\big) q''\big(\theta^{*}+
		\tilde{\theta}(s)\big)\dot{\tilde{\theta}}(s){\rm{d}}s{\rm{d}}\tau\right|\\  \ \ \ \ \ \  \ \ \ \ \
		\leq \frac{1}{\varepsilon}\int_{t-\varepsilon}^{t}
		\int_{\tau-h(\tau)}^{t-D}{\rm{d}}s{\rm{d}}\tau\cdot \left|q''\big(\theta^{*}+
		\tilde{\theta}(s)\big)\right| \cdot \left|\dot{\tilde{\theta}}(s)\right|\\  \ \ \ \ \ \  \ \ \ \ \
		<\frac{1}{2}{q_2(\sigma)\Delta_\theta}(\varepsilon+2\bar{d}),\\
		\left|\dot z(t)\right|\leq \left|kq'\big(\theta^{*}+z(t-D)+kG(t-D)\big)\right| +\left|\frac{2k}{a}Y_1(t)\right|+\left|2kY_2(t)\right|+\left|ka\mathcal{R}(t)\right| {+|k\mathcal{W}(t)|}\\  \ \ \ \ \ \  \ \ \ 
		<k\big(q_1(\sigma)+\frac{1}{a}{q_1(\sigma)\Delta_\theta}(\varepsilon+2\bar{d})  +{q_2(\sigma)\Delta_\theta}(\varepsilon+2\bar{d})+a
		\Delta_\mathcal{R} {+\frac{2}{a}\nu^*}\big)\triangleq \Delta_z.
	\end{array}
\end{equation}
Besides, through the differential mean value theorem for \eqref{1tvr_define_Phi}, and under the bounds (\ref{1tvr_AssumptionBound2}) in Assumption 2 and (\ref{1tvr_boundsGY}), we arrive at
\begin{equation}\label{1tvr_boundphi}
	\begin{array}{l}
		|\Phi(t)|\leq \left|q''\big(\theta^{*}+z(t-D)+\zeta kG(t-D)\big)\right|\cdot|kG(t-D)|  
		<q_2(\sigma)\frac{k\Delta_\mathcal{F}}{2}\varepsilon,
	\end{array}
\end{equation}
where $\zeta\in(0,1)$.	

\begin{theorem}\label{theorem1} Under Assumptions 1-2 with a given $\sigma$, consider the closed-loop system consisting of the scalar map (\ref{1tvr_map}) and the ES controller (\ref{1tvr_controller}), as well as the initial condition $\left|{\tilde{\theta}}(t)\right|\leq \sigma_0<\sigma$ for $t\in\left[0,\bar{h}\right]$. Given {$D\geq 0$} and tuning parameters $\sigma_0, k, a, \delta, \bar{d}^*, \varepsilon^*>0$, and {$\nu^*\geq 0$}, let scalar decision variables $\lambda_1, \lambda_2, \lambda_3, \lambda_4, \lambda_5, {\gamma
>0}$, and $P\geq
1, Q>0, R>0$ satisfy the following LMIs:
\begin{equation}\label{1tvr_lmi1}
	\begin{array}{l}
		\Omega_0=\begin{bmatrix}
			\Psi    &\Xi\\
			*       &-R
		\end{bmatrix}< 0,\ \ \ \ \ \ \Omega_1=\frac{\varpi}{2\delta}<\big(\sigma-\frac{k\Delta_\mathcal{F}}{2}\varepsilon^*\big)^2,\\
		\Omega_2=(P+{Q}D)\varkappa^2+RD{{}^3}\Delta_z^2<\big(\sigma-\frac{k\Delta_\mathcal{F}}{2}\varepsilon^*\big)^2,\\
		
	\end{array}
\end{equation}
where $\Psi$ is the symmetric matrix composed of
\begin{equation}
	\begin{array}{l}\label{1tvr_Psi}
		\Psi_{11}=-\big(2k\mu(\sigma)-2\delta\big) P+Q, \ \ \ \ 
		\Psi_{13}=kL_1P, \\
		\Psi_{15}=kP, \ \ \ \  \Psi_{16}=\frac{2k}{a}P,\ \ \ \ \  \Psi_{17}=2kP,
		\\
		\Psi_{18}=kaP,\ \ \  {\Psi_{19}=kP,}  \ \ \ \ \
		\Psi_{22}=-Qe^{-2\delta D}+\lambda_1 L_1^2,\\
		\Psi_{33}=-Re^{-2\delta D},\ \ \ \ 
		\Psi_{44}=-\lambda_1,\ \ \ \ 
		\Psi_{55}=-\frac{\lambda_2}{\varepsilon^*},\\
		\Psi_{66}=-\frac{\lambda_3}{\varepsilon^*+2\bar{d}^*},\ \ \ \ 
		\Psi_{77}=-\frac{\lambda_4}{\varepsilon^*+2\bar{d}^*},\ \ \ \ 
		\Psi_{88}=-\lambda_5a, \ \ \ \ {\Psi_{99}=-\gamma,} 
	\end{array}
\end{equation}
with other terms being zero, and
\begin{equation}\label{1tvr_Xi}
	\begin{array}{l}
		\Xi=\left[0,0,0,k,k,\frac{2k}{a},2k,ka, {k}\right]^{\rm T}DR,\\
		\varkappa=\sigma_0+\Delta_\theta \varepsilon^*+\frac{k \Delta_\mathcal{F}}{2}\varepsilon^*,\ \
		\\
		\varpi=\frac{\lambda_2q_2^2(\sigma)k^2 \Delta_\mathcal{F}^2}{4}\varepsilon^*+\frac{\lambda_3 q_1^2(\sigma)\Delta_\theta^2}{4}(\varepsilon^*+2\bar{d}^*)
		 +\frac{\lambda_4 q_2^2(\sigma)\Delta_\theta^2}{4}(\varepsilon^*+2\bar{d}^*)+\lambda_5a \Delta_\mathcal{R}^2 {+\frac{4\gamma}{a^2}{\nu^*}^2}.
	\end{array}
\end{equation}
Then, for $\forall$ $\bar{d} \in [0, \bar{d}^*]$ and $\forall$ $\varepsilon \in (4\bar{d}, \varepsilon^*]$, the estimation error satisfies
\begin{equation}\label{1tvr_theorem1boundtheta}
	\begin{array}{l}
		\Big|{\tilde{\theta}}(t)\Big|< \left|{\tilde{\theta}}(\bar{h})\right|+\Delta_\theta (t-\bar{h})<\sigma, \ \
		t\in \left[\bar{h}, \varepsilon+\bar{h}\right],\\
		\left|{\tilde{\theta}}(t)\right|<  \left[\big((P+{Q}D)\varkappa^2+RD^3\Delta_z^2\big)e^{-2\delta(t-\varepsilon-\bar{h})}\big(1-e^{-2\delta(t-\varepsilon-\bar{h})}\big)\frac{\varpi}{2\delta}\right]^{\frac{1}{2}} +\frac{k\Delta_\mathcal{F}}{2}\varepsilon<\sigma,\ \ \ t\in \left[\varepsilon+\bar{h}, \infty \right),
	\end{array}
\end{equation}
and is exponentially attracted to the set
\begin{equation}\label{ubscalar}
	\begin{array}{l}
		\Theta=\left\{\tilde{\theta} \in \mathbb{R}: \left|{\tilde{\theta}}\right|< \sqrt{\frac{\varpi}{2\delta}}+\frac{k\Delta_\mathcal{F}}{2}\varepsilon \right\},
	\end{array}
\end{equation}
which is adjustable via the tuning parameters $\varepsilon$ and $a$. Moreover, LMIs \eqref{1tvr_lmi1} are always feasible for small enough $\bar{d}^*,\varepsilon^*, {\nu^*}$, $a$, and $k$. 
\end{theorem}
\textit{Proof}: The proof of Theorem 1 follows the same argument of the proof of Theorem 2 which is given in Appendix.

Note that the exponential decay rate of $\left|\tilde{\theta}(t)\right|$ in \eqref{1tvr_theorem1boundtheta} is $\delta$, which is governed by the adaptation gain $k$ such that $0<\delta<k\mu(\sigma)$ in $\Omega_0$ of \eqref{1tvr_lmi1}.

\section{Extremum Seeking of Vector Case}
In this section, {building upon the scalar analysis framework, we extend the methodology to vector case}, which is non-trivial. 
We consider the ES for nonlinear static maps
\begin{equation}\label{2tvr_map}
	y(t)=q\big(\theta(t)\big),
\end{equation}
where $y(t)\in \mathbb{R}$ is the measurable output, {$\theta(t)=[\theta_1(t), \cdots,  \theta_n(t)]^{\rm T}\in \mathbb{R}^n$} is the input vector. As shown in Fig. {\ref{figtvr_vector}}, {$\nu(t)\in \mathbb{R}$ represents the uncertainty in measurement bias and satisfies the condition given in \eqref{bound_nu}}. The ES scheme is subject to transmission delays in both input and output channels, which have been defined by \eqref{tvrdelay}-\eqref{tvrdelayadded}. The map (\ref{2tvr_map}) satisfies the following assumptions.
\begin{figure}[pos = H]
	\centering
	\includegraphics[scale=0.32]{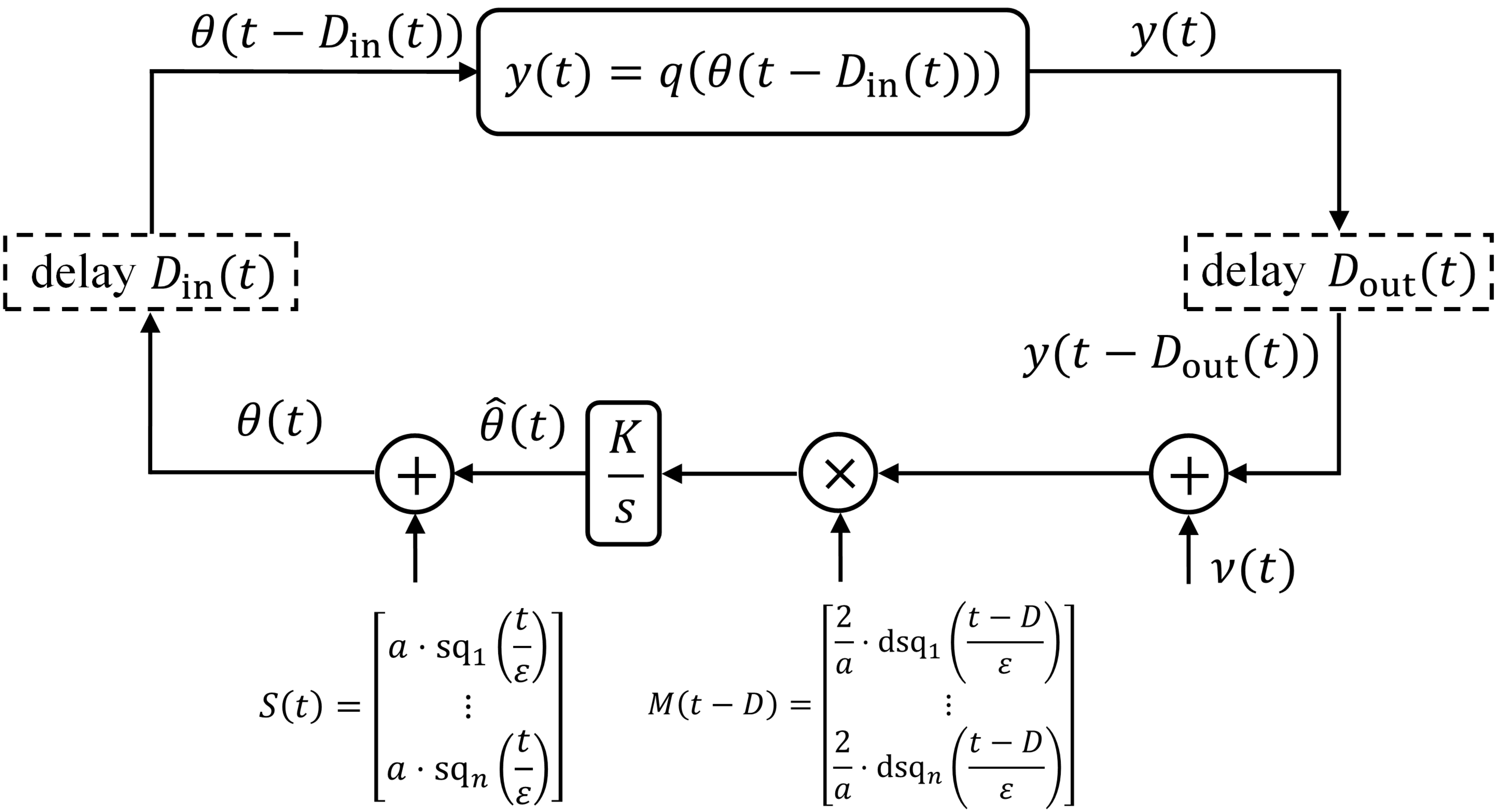}
	\caption{{ES with {a disturbance and} time delays for static maps: vector case.}}
	\label{figtvr_vector}
\end{figure}

\begin{assump} \cite{gaofeng2023}: There exist $\theta^*=[\theta^*_1, \cdots, \theta^*_n]^{\rm T}\in \mathbb{R}^n$, constants $\sigma>0$ and small $a>0$ such that $q(\theta)=q(\theta_1,\cdots, \theta_n)\in C^2 \left[(\theta^*_1-\sigma-a, \theta^*_1+\sigma+a){\small\times\cdots \times}(\theta^*_n-\sigma-a, \theta^*_n+\sigma+a)\right]$, and the following relations hold:
	\begin{equation}\label{2tvr_defineH}
				\begin{array}{l}
					\frac{\partial q}{\partial \theta}(\theta^*)=\left[\frac{\partial q}{\partial \theta_1}(\theta^*) \  \cdots \ \frac{\partial q}{\partial \theta_n}(\theta^*)\right]=0,\ \\ [5pt]
					\frac{\partial^2 q}{\partial \theta^2}(\theta^*)= \begin{bmatrix}
						\frac{\partial^2 q}{\partial \theta^2_1}(\theta^*)&\cdots &\frac{\partial^2 q}{\partial \theta_1 \partial \theta_n }(\theta^*)\\
						\vdots   &\ddots & \vdots\\
						\frac{\partial^2 q}{\partial \theta_n \partial \theta_1 }(\theta^*)&\cdots&\frac{\partial^2 q}{\partial \theta^2_n}(\theta^*)
					\end{bmatrix}
					=H<0,
			\end{array}
	\end{equation}
	\begin{equation}\label{2tvr_boundsAssumption}
		\begin{array}{l}
			\frac{\partial q}{\partial \theta}(\theta^*+\Delta)\cdot \Delta \leq -\mu(\sigma)\cdot |\Delta|^2<0,\ \
			\Delta =\left[\Delta_1, \cdots, \Delta_n \right]^{\rm T},  \  \forall \ 0<|\Delta_i|<\sigma,\ \  i=1,  \cdots, n,
		\end{array}
\end{equation}
where $\mu(\sigma)>0$ is a known $\sigma$--dependent constant, which decreases monotonically in $\sigma$. 
\end{assump}

\begin{assump} \cite{gaofeng2023}: For any $\Delta =\left[\Delta_1, \cdots, \Delta_n \right]^{\rm T}$ with $0<|\Delta_i|<\sigma,\   i=1, \cdots, n$, and $a$ defined in Assumption 3, given $\xi=\left[\xi_1, \cdots, \xi_n \right]^{\rm T}$ with $\xi_i \in[-1, 1], \ i=1, \cdots, n$, we have
\begin{equation}\label{2tvr_AssumptionBound2}
	\begin{array}{l}
		|q(\theta^*+\Delta)|<q_0(\sigma),\ \ \big|\frac{\partial q}{\partial \theta}(\theta^*+\Delta)\big|<q_1(\sigma),\ \
		\big|\frac{\partial^2 q}{\partial \theta^2}(\theta^*+\Delta)\big|<q_2(\sigma),\\  \big|\frac{\partial q}{\partial \theta}(\theta^*+\Delta)-\frac{\partial q}{\partial \theta}(\theta^*)\big|<L_1|\Delta|,\ \
		\big|\frac{\partial^2 q}{\partial \theta^2}(\theta^*+\Delta+a\xi)-\frac{\partial^2 q}{\partial \theta^2}(\theta^*)\big|<L_2|\Delta+a\xi|,
	\end{array}
\end{equation}
in which $q_0(\sigma), q_1(\sigma), q_2(\sigma), L_1$, and $L_2$ are positive constants. 
\end{assump}

The gradient-based ES regulator is given by
\begin{equation}\label{2tvr_controller}
	\begin{array}{l}
		{\theta}(t)=\hat{\theta}(t)+S\left(t\right),\\
		\dot{\hat{\theta}}(t)=K\cdot M(t-D)\cdot \left[y\left(t-D_{\rm{out}}(t)\right){{+\nu(t)}}\right]\\ \ \ \ \ \ \ \   =K\cdot M(t-D)\cdot \left[q\big(\theta\left(t-h(t)\right)\big){{+\nu(t)}}\right],\ \ \ \  t\ge\bar{h},
	\end{array}
\end{equation}
where $\hat{\theta}(t)$ is the real-time estimate of $\theta^*$ with the estimation error being
\begin{equation}\label{2tvr_define_esterror}
	\tilde{\theta}(t)\triangleq\hat{\theta}(t)-\theta^*.
\end{equation}
The initial value satisfies $\left|\hat{\theta}(t)-\theta^*\right|\le\sigma_0<\sigma$ for $t\in\left[0,\bar{h}\right]$ where $\sigma_0>0$ is a known constant. {The matrix $K=kI$ with $k>0$ being the adaptation gain whose sign is opposite to the sign of the Hessian matrix $H$ in (\ref{2tvr_defineH}) and $I$ being the identity matrix.} The dither signals are defined as
{\begin{equation*}
		\begin{array}{l}
			S(t)= a\left[
			{\rm{sq}}_1\left(\frac{t}{\varepsilon}\right), \cdots, {\rm{sq}}_n\left(\frac{t}{\varepsilon}\right)
			\right]^{\rm T},\\
			M(t-D)= \frac{2}{a}\left[
			{\rm{dsq}}_1\left(\frac{t-D}{\varepsilon}\right), \cdots, {\rm{dsq}}_n\left(\frac{t-D}{\varepsilon}\right)
			\right]^{\rm T},
		\end{array}
\end{equation*}}
with the dither amplitude $a>0$ and the dither period $\varepsilon>0$. {The signals ${\rm{sq}}_{\ell}\left(\frac{t}{\varepsilon}\right)$ and ${\rm{dsq}}_{\ell}\left(\frac{t}{\varepsilon}\right)$ are designed, for $\ell=1, \cdots, n$, as follows:
	\begin{equation*}
				\begin{aligned}
					\mathrm{sq}_{\ell}\left(\frac{t}{\varepsilon}\right)=
					\begin{cases}
						1, 
						& t \in 
						\displaystyle 
						\ \bigcup_{m=0}^{2^{\ell-1}-1}
						\varepsilon\!\left[
						i+\frac{m\cdot 2}{2^{\ell}},\ 
						i+\frac{1+m\cdot 2}{2^{\ell}}
						\right),\\[1.0em]
						-1, 
						& t \in 
						\displaystyle 
						\ \bigcup_{m=0}^{2^{\ell-1}-1}
						\varepsilon\!\left[
						i+\frac{1+m\cdot 2}{2^{\ell}},\ 
						i+\frac{2+m\cdot 2}{2^{\ell}}
						\right),
					\end{cases},i\in \mathbb{N},
			\end{aligned}
	\end{equation*}
	\begin{equation*}
				\begin{aligned}
					\mathrm{dsq}_{\ell}\left(\frac{t}{\varepsilon}\right)=
					\begin{cases}
						1, 
						& t \in 
						\displaystyle\bigcup_{j=0}^{2^{n-\ell}-1} 
						\ \bigcup_{m=0}^{2^{\ell-1}-1}
						\varepsilon\!\left[
						i+\frac{1+2j+m\cdot 2^{n-\ell+2}}{2^{n+1}},\ 
						i+\frac{2+2j+m\cdot 2^{n-\ell+2}}{2^{n+1}}
						\right),\\[1.0em]
						-1, 
						& t \in 
						\displaystyle\bigcup_{j=0}^{2^{n-\ell}-1} 
						\ \bigcup_{m=0}^{2^{\ell-1}-1}
						\varepsilon\!\left[
						i+\frac{1+2j+m\cdot 2^{n-\ell+2}}{2^{n+1}}+\frac{1}{2^{\ell}},\ 
						i+\frac{2+2j+m\cdot 2^{n-\ell+2}}{2^{n+1}}+\frac{1}{2^{\ell}}
						\right),\\[1.0em]
						0, & \text{others},  
					\end{cases}, i\in \mathbb{N}.
			\end{aligned}
\end{equation*}}
\begin{figure}[pos=H]
	\begin{center}
		\includegraphics[scale=0.32]{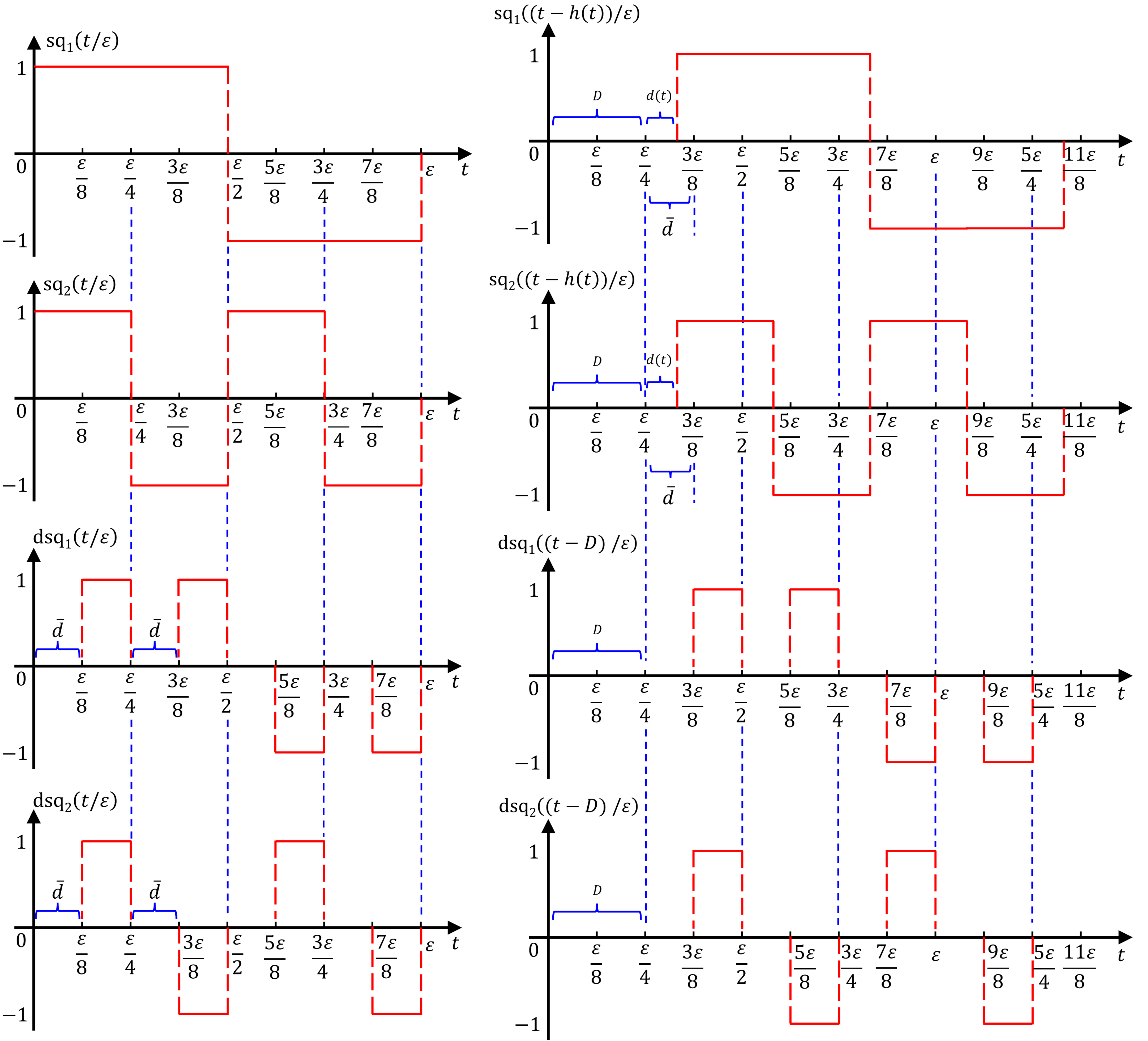}
		\DeclareGraphicsExtensions.
		\caption{The dithers and the delayed dithers under $n=2$, $D=\frac{\varepsilon}{4}$, and $\bar{d}\leq\frac{\varepsilon}{8}$.}
		\label{2tvr_dither_and_delayed}
	\end{center}
\end{figure}

{The case $n=2$ is shown in Fig. \ref{2tvr_dither_and_delayed}. In parallel with the scalar case (see the illustration underneath \eqref{1tvr_define_dithersignal2}), the phase-shifted dither $M(t-D)$ is employed to address the constant delay $D$, and the signals ${\rm{dsq}}_{\ell}\left(\frac{t}{\varepsilon}\right), \ell=1, \cdots, n$, are carefully designed to deal with the unknown time-varying delay $d(t)$ in (\ref{tvrdelay})-\eqref{tvrdelayub}. When $\bar{d}$ is estimated and known, we are allowed to select $\varepsilon>2^n\bar{d}$ to establish \eqref{2averaging} so that the delay-free (when $h(t)\equiv0$) averaged system \eqref{2averaged} is stable. For computation simplicity, in Section III of this paper, we consider the case of $\varepsilon\geq2^{n+1}\bar{d}$.}

Under \eqref{2tvr_controller}--\eqref{2tvr_define_esterror}, utilizing the Taylor series of $q\left(\cdot\right)$, the dynamics of the estimation error is govern by
\begin{equation}\label{dee}
	\begin{array}{l}
		\dot{\tilde{\theta}}(t)=\dot{\hat{\theta}}(t)=KM(t-D) \left[q\big(\theta\left(t-h(t)\right)\big){{+\nu(t)}}\right]\\
		=KM(t-D)\left[q\big(\theta^{*}+\tilde{\theta}(t-h(t))+S(t-h(t))\big){{+\nu(t)}}\right]\\
		=K M(t-D)\Big[q\left(\theta^{*}+\tilde{\theta}\left(t-h(t)\right)\right) +\frac{\partial q}{\partial \theta}\left(\theta^{*}+\tilde{\theta}\left(t-h(t)\right)\right)S\left(t-h(t)\right)\\
		\ \ \ \ +\frac{1}{2}S^{\rm T}\left(t-h(t)\right)H(t)S\left(t-h(t)\right) {{+\nu(t)}}\Big],\ \ \ \ t\ge\bar{h},
	\end{array}
\end{equation}
where $H(t)\triangleq\frac{\partial ^2 q}{\partial \theta ^2}\big(\theta^{*}+\tilde{\theta}(t-h(t))+\zeta S(t-h(t))\big)$ with $\zeta\in(0,1)$. Denoting the notation
{\begin{equation*}
		\begin{array}{l}
			S(t-h(t))=aS_0(t-h(t)), \\
			S_0(t-h(t))\triangleq \left[
			{\rm{sq}}_1\big(\frac{t-h(t)}{\varepsilon}\big), \cdots,  {\rm{sq}}_n\big(\frac{t-h(t)}{\varepsilon}\big)
			\right]^{\rm{T}}, \\
			M(t-D)=\frac{2}{a}M_0(t-D), \\
			M_0(t-D)\triangleq \left[
			{\rm{dsq}}_1\big(\frac{t-D}{\varepsilon}\big), \cdots,  {\rm{dsq}}_n\big(\frac{t-D}{\varepsilon}\big)
			\right]^{\rm{T}},
		\end{array}
\end{equation*}}
the system \eqref{dee} is rewritten as
\begin{equation}\label{2tvr_errorSystem0}
	\begin{array}{l}
		\dot{\tilde{\theta}}(t)
		=\frac{2k}{a}M_0(t-D)q\left(\theta^{*}+\tilde{\theta}(t-h(t))\right)\\
		\ \ \ \ \ \ \ \ \ \ +2kM_0(t-D)\frac{\partial q}{\partial \theta}\left(\theta^{*}+\tilde{\theta}(t-h(t))\right)S_0(t-h(t))\\
		\  \ \ \ \ \ \ \ \ \ +kaM_0(t-D)S_0^{\rm T}(t-h(t))H(t)S_0(t-h(t)) {{+\frac{2k}{a}\nu(t)M_0(t-D)}}.\\
	\end{array}
\end{equation}
Making use of $\frac{\partial q}{\partial\theta}\left(\theta^{*}+\tilde{\theta}(t-h(t))\right)S_0(t-h(t))=S^{\rm T}_0(t-h(t))\frac{\partial q}{\partial \theta}^{\rm T}\left(\theta^{*}+\tilde{\theta}(t-h(t))\right)$, the equation \eqref{2tvr_errorSystem0} can be further expressed in the following form
\begin{equation}\label{2tvr_errorSystemFR}
	\dot{\tilde{\theta}}(t)
	=k\mathcal{F}(t)+ka\mathcal{R}(t) {+k\mathcal{W}(t)},\ \ t\ge\bar{h},
\end{equation}
where
\begin{equation}\label{2tvr_DefineFR}
	\begin{array}{l}
		\mathcal{F}(t)\triangleq\frac{2}{a}M_0(t-D)q\left(\theta^{*}+\tilde{\theta}\left(t-h(t)\right)\right)\\
		\ \ \ \ \ \ \ \ \ \ \ \ +2M_0(t-D)S^{\rm T}_0\left(t-h(t)\right)\frac{\partial q}{\partial \theta}^{\rm T}\left(\theta^{*}+\tilde{\theta}(t-h(t))\right)\\
		\ \ \ \ \ \ \ \ \ \ \ \ +aM_0(t-D)S_0^{\rm T}\left(t-h(t)\right)HS_0\left(t-h(t)\right),\\
		\mathcal{R}(t)\triangleq M_0(t-D)S_0^{\rm T}\left(t-h(t)\right)(H(t)-H)S_0\left(t-h(t)\right),\\
		{{\mathcal{W}(t)\triangleq\frac{2}{a}\nu(t)M_0(t-D)}}.
	\end{array}
\end{equation}
Similar to the scalar case \eqref{1averagedsystem}-\eqref{1tvr_integerdsq}, regarding slowly time-varying $q\left(\theta^{*}+\tilde{\theta}\left(t-h(t)\right)\right)$ and $\frac{\partial q}{\partial \theta}^{\rm T}\left(\theta^{*}+\tilde{\theta}(t-h(t))\right)$ as freezing constants, and neglecting $ka\mathcal{R}(t)$, the averaged-like system of \eqref{2tvr_errorSystemFR} is calculated as
\begin{equation}\label{2averaged}
	\begin{array}{l}
		\dot{\tilde{\theta}}_{av}(t)=\frac{2k}{a\varepsilon}\int_{t-\varepsilon}^tM_0(\tau-D)d\tau\cdot q\left(\theta^{*}+\tilde{\theta}_{av}\left(t-h(t)\right)\right)\\
		\ \ \ \ \ \ \ \ \ \ \ \ \ +\frac{2k}{\varepsilon}\int_{t-\varepsilon}^tM_0(\tau-D)S^{\rm T}_0\left(\tau-h(\tau)\right)d\tau\frac{\partial q}{\partial \theta}^{\rm T}\left(\theta^{*}+\tilde{\theta}_{av}(t-h(t))\right)\\
		\ \ \ \ \ \ \ \ \ \ \ \ \ +\frac{ak}{\varepsilon}\int_{t-\varepsilon}^tM_0(\tau-D)S^{\rm T}_0\left(\tau-h(\tau)\right)HS_0\left(\tau-h(\tau)\right)d\tau {+k\mathcal{W}(t)}\\ \ \ \ \ \ \ \ \ \ \ 
		=k\frac{\partial q}{\partial \theta}^{\rm T}\left(\theta^{*}+\tilde{\theta}_{av}(t-h(t))\right) {+k\mathcal{W}(t)},
	\end{array}
\end{equation}
where we utilize
{\begin{equation}\label{2averaging}
				\begin{array}{l}
					\frac{1}{\varepsilon}\int_{t-\varepsilon}^tM_0(\tau-D)d\tau=\frac{1}{\varepsilon}\int_{t-\varepsilon}^t\left[\begin{smallmatrix}	{\rm{dsq}}_1\left(\frac{\tau-D}{\varepsilon}\right)\\ \vdots \\ {\rm{dsq}}_n\left(\frac{\tau-D}{\varepsilon}\right)\end{smallmatrix}\right]d\tau=0,\\
					\frac{1}{\varepsilon}\int_{t-\varepsilon}^tM_0(\tau-D)S^{\rm T}_0\left(\tau-h(\tau)\right)d\tau\\
					=\frac{1}{\varepsilon}\int_{t-\varepsilon}^t\left[\begin{smallmatrix}	{\rm{dsq}_1}\left(\frac{\tau-D}{\varepsilon}\right){\rm{sq}}_1\left(\frac{\tau-h(\tau)}{\varepsilon}\right)&\cdots&{\rm{dsq}_1}\left(\frac{\tau-D}{\varepsilon}\right){\rm{sq}}_n\left(\frac{\tau-h(\tau)}{\varepsilon}\right)\\
						\vdots &\ddots	&\vdots \\
						{\rm{dsq}}_n\left(\frac{\tau-D}{\varepsilon}\right){\rm{sq}}_1\left(\frac{\tau-h(\tau)}{\varepsilon}\right)&\cdots&{\rm{dsq}}_n\left(\frac{\tau-D}{\varepsilon}\right){\rm{sq}}_n\left(\frac{\tau-h(\tau)}{\varepsilon}\right)\end{smallmatrix}\right] d\tau=\frac{1}{2}I,\\
					\frac{1}{\varepsilon}\int_{t-\varepsilon}^tM_0(\tau-D)S^{\rm T}_0\left(\tau-h(\tau)\right)HS_0\left(\tau-h(\tau)\right)d\tau=0.
				\end{array}
\end{equation}}

Subsequently, we employ the time-delay approach to analyze the ES system described in \eqref{2tvr_errorSystemFR}
\begin{equation}\label{2tvr_integralF}
	\begin{array}{l}
		\frac{1}{\varepsilon}\int_{t-\varepsilon}^{t}\mathcal{F}(\tau){\rm{d}}\tau\\
		=\frac{2}{a\varepsilon}
		\int_{t-\varepsilon}^{t}M_0(\tau-D)q\big(\theta^{*}+\tilde{\theta}(\tau-h(\tau))\big){\rm{d}}\tau\\
		\ \ \ +\frac{2}{\varepsilon}\int_{t-\varepsilon}^{t}M_0(\tau-D)S^{\rm T}_0(\tau-h(\tau)) \frac{\partial q}{\partial \theta}^{\rm T}\left(\theta^{*}+\tilde{\theta}(\tau-h(\tau))\right){\rm{d}}\tau, \ \ t\ge\varepsilon+\bar{h},
	\end{array}
\end{equation}
in which the 3rd formula of \eqref{2averaging} is employed.
Considering the zero equality $\frac{1}{\varepsilon}\int_{t-\varepsilon}^{t}M_0(\tau-D)q\big(\theta^{*}+\tilde{\theta}(t-D)\big){\rm{d}}\tau=0$, for the 1st term on the right-hand side of (\ref{2tvr_integralF}) , we arrive at
\begin{equation}\label{2tvr_integralF_1}
	\begin{array}{l}
		\frac{2}{a\varepsilon}
		\int_{t-\varepsilon}^{t}M_0(\tau-D)q\big(\theta^{*}+\tilde{\theta}(\tau-h(\tau))\big){\rm{d}}\tau\\
		=-\frac{2}{a\varepsilon}
		\int_{t-\varepsilon}^{t}M_0(\tau-D)\left[q\big(\theta^{*}+\tilde{\theta}(t-D)\big)-q\big(\theta^{*}+\tilde{\theta}(\tau-h(\tau))\big)\right]{\rm{d}}\tau\\
		=-\frac{2}{a\varepsilon}\int_{t-\varepsilon}^{t}M_0(\tau-D)\int_{\tau-h(\tau)}^{t-D}\frac{\partial q}{\partial \theta}\left(\theta^{*}+\tilde{\theta}(s)\right)\dot{\tilde{\theta}}(s){\rm{d}}s{\rm{d}}\tau.
	\end{array}
\end{equation}
For the 2nd term on the right-hand side of (\ref{2tvr_integralF}) , we have
\begin{equation}\label{2tvr_integralF_2}
	\begin{array}{l}
		\frac{2}{\varepsilon}\int_{t-\varepsilon}^{t}M_0(\tau-D)S^{\rm T}_0(\tau-h(\tau))\frac{\partial q}{\partial \theta}^{\rm T}\Big(\theta^{*}+\tilde{\theta}(\tau-h(\tau))\Big){\rm{d}}\tau\\
		=\frac{2}{\varepsilon}\int_{t-\varepsilon}^{t}M_0(\tau-D)S^{\rm T}_0(\tau-h(\tau))\\
		\ \ \ \times\left[\frac{\partial q}{\partial \theta}^{\rm T}\left(\theta^{*}+\tilde{\theta}(\tau-h(\tau))\right) - \frac{\partial q}{\partial \theta}^{\rm T}\left(\theta^{*}+\tilde{\theta}(t-D)\right)+\frac{\partial q}{\partial \theta}^{\rm T}\left(\theta^{*}+\tilde{\theta}(t-D)\right)\right]{\rm{d}}\tau\\
		=\frac{2}{\varepsilon}\int_{t-\varepsilon}^{t}M_0(\tau-D)S^{\rm T}_0(\tau-h(\tau)){\rm{d}}\tau\cdot \frac{\partial q}{\partial \theta}^{\rm T}\left(\theta^{*}+\tilde{\theta}(t-D)\right)\\ \ \ \
		-\frac{2}{\varepsilon}\int_{t-\varepsilon}^{t}M_0(\tau-D)S^{\rm T}_0(\tau-h(\tau))
		\left[\frac{\partial q}{\partial \theta}^{\rm T}\left(\theta^{*}+\tilde{\theta}(t-D)\right)-\frac{\partial q}{\partial \theta}^{\rm T}\left(\theta^{*}+\tilde{\theta}(\tau-h(\tau))\right)\right]{\rm{d}}\tau\\
		=\frac{\partial q}{\partial \theta}^{\rm T}\left(\theta^{*}+\tilde{\theta}(t-D)\right) -\frac{2}{\varepsilon}\int_{t-\varepsilon}^{t}M_0(\tau-D)S^{\rm T}_0(\tau-h(\tau))\int_{\tau-h(\tau)}^{t-D}\frac{\partial^2 q}{\partial \theta^2}\left(\theta^{*}+\tilde{\theta}(s)\right)\dot{\tilde{\theta}}(s){\rm{d}}s{\rm{d}}\tau,
	\end{array}
\end{equation}
where we employed $\frac{2}{\varepsilon}\int_{t-\varepsilon}^{t}M_0(\tau-D)S^{\rm T}_0(\tau-h(\tau)){\rm{d}}\tau=I$, $I$ is an identity matrix. We define
\begin{equation}\label{2tvr_DefineG}
	\begin{array}{l}
		G(t)=\frac{1}{\varepsilon}
		\int_{t-\varepsilon}^{t}(\tau-t+\varepsilon)\mathcal{F}(\tau){\rm{d}}\tau,
	\end{array}
\end{equation}
and then the following relation is satisfied
\begin{equation}\label{2tvr_systemFF}
	\begin{array}{l}
		\frac{\rm{d}}{{\rm{d}}t}\left[{\tilde{\theta}}(t)-kG(t)\right]=\dot{\tilde{\theta}}(t)-k\mathcal{F}(t)+\frac{k}{\varepsilon}
		\int_{t-\varepsilon}^{t}\mathcal{F}(\tau){\rm{d}}\tau.
	\end{array}
\end{equation}

Combining \eqref{2tvr_errorSystemFR}, \eqref{2tvr_integralF}--\eqref{2tvr_integralF_2}, and \eqref{2tvr_systemFF} together, we present the closed-loop error system as follows:
\begin{equation}\label{2tvr_errorSystemThetaYR}
	\begin{array}{l}
		\frac{\rm{d}}{{\rm{d}}t}\left[{\tilde{\theta}}(t)-kG(t)\right]=k\frac{\partial q}{\partial \theta}^{\rm T}\left(\theta^{*}+\tilde{\theta}(t-D)\right)
		 -\frac{2k}{a}Y_1(t) -2kY_2(t)+ka\mathcal{R}(t) {+k\mathcal{W}(t)},\ \ \  t\geq \varepsilon+\bar{h},
	\end{array}
\end{equation}
where
\begin{equation}\label{2tvr_DefineY}
	\begin{array}{l}
		Y_1(t)\triangleq\frac{1}{\varepsilon}\int_{t-\varepsilon}^{t}\int_{\tau-h(\tau)}^{t-D}M_0(\tau-D)\frac{\partial q}{\partial \theta}\left(\theta^{*}+\tilde{\theta}(s)\right)\dot{\tilde{\theta}}(s){\rm{d}}s{\rm{d}}\tau,\\
		Y_2(t)\triangleq\frac{1}{\varepsilon}\int_{t-\varepsilon}^{t}\int_{\tau-h(\tau)}^{t-D}M_0(\tau-D)S^{\rm T}_0(\tau-h(\tau))\frac{\partial^2 q}{\partial \theta^2}\left(\theta^{*}+\tilde{\theta}(s)\right)\dot{\tilde{\theta}}(s){\rm{d}}s{\rm{d}}\tau,
	\end{array}
\end{equation}
with $\dot{\tilde{\theta}}(t)$ defined by (\ref{2tvr_errorSystemFR}). We further employ the following change of variable
\begin{equation}\label{2tvr_denoteZ}
	\begin{array}{l}
		z(t)=\tilde{\theta}(t)-kG(t),
	\end{array}
\end{equation}
then the closed-loop system (\ref{2tvr_errorSystemThetaYR}) is transformed into
\begin{equation}\label{2tvr_errorSystemZYR}
	\begin{array}{l}
		\dot z(t)=k\frac{\partial q}{\partial \theta}^{\rm T}\Big(\theta^{*}+z(t-D)+kG(t-D)\Big)
		 -\frac{2k}{a}Y_1(t)	-2kY_2(t)+ka\mathcal{R}(t) {+k\mathcal{W}(t)}\\ \ \ \ \ \ \ \  
		=k\frac{\partial q}{\partial \theta}^{\rm T}\Big(\theta^{*}+z(t-D)\Big)+k\Phi(t)
		-\frac{2k}{a}Y_1(t)  -2kY_2(t)+ka\mathcal{R}(t) {+k\mathcal{W}(t)},\ \  t\geq \varepsilon+\bar{h},
	\end{array}
\end{equation}
where
\begin{equation}
	\begin{array}{l}\label{2tvr_define_Phi}
		\Phi(t)\triangleq\frac{\partial q}{\partial \theta}^{\rm T}\Big(\theta^{*}+z(t-D)+kG(t-D)\Big) -\frac{\partial q}{\partial \theta}^{\rm T}\Big(\theta^{*}+z(t-D)\Big).
	\end{array}
\end{equation}

We postulate that
\begin{equation}\label{2tvr_overallBound}
	\begin{array}{l}
		\left|\tilde{\theta}(t)\right|<\sigma,\ \  t\geq 0,
	\end{array}
\end{equation}
which will be ensured by the conditions \eqref{2tvr_lmi1} in Theorem 2. By defining $\overline{S}_0=\sup_{t\geq 0}\left|S_0(t)\right|, \overline{M}_0=\sup_{t\geq 0}\left|M_0(t)\right|$, under the bounds (\ref{2tvr_AssumptionBound2}) in Assumption 4 and the overall bound (\ref{2tvr_overallBound}), the upper bounds on $\mathcal{F}(t), \mathcal{R}(t)$, and $\dot{\tilde{\theta}}(t)$ are obtained from (\ref{2tvr_DefineFR}) and (\ref{2tvr_errorSystemFR}) such that
\begin{equation}\label{2tvr_boundFR}
	\begin{array}{l}
		|\mathcal{F}(t)|\leq \frac{2}{a}\left|M_0(t-D)q\big(\theta^{*}+\tilde{\theta}(t-h(t))\big)\right|
		\\ \ \  \ \  \ \  \ \  \ \  \ \ \ +2\left|M_0(t-D)S^{\rm T}_0(t-h(t))\frac{\partial q}{\partial \theta}^{\rm T}\big(\theta^{*}+\tilde{\theta}(t-h(t))\big)\right|\\  \ \  \ \  \ \  \ \  \ \  \ \ \ +a\left|M_0(t-D)S_0^{\rm T}\left(t-h(t)\right)HS_0\left(t-h(t)\right)\right|
		\\  \ \  \ \  \ \  \ \  \ \  \ 
		<\frac{2}{a}\overline{M}_0q_0(\sigma)+2\overline{M}_0\overline{S}_0q_1(\sigma)+a\overline{M}_0\overline{S}_0^2q_2(\sigma)\triangleq \Delta_\mathcal{F},\\
		|\mathcal{R}(t)|=  \left|M_0(t-D)S_0^{\rm T}(t-h(t))(H(t)-H)S_0(t-h(t))\right|\\   \ \  \ \  \ \  \ \  \ \ \   
		<\overline{M}_0\overline{S}_0^2 L_2\left|\tilde{\theta}(t-h(t))+a\zeta S_0(t-h(t))\right|\\  \ \  \ \  \ \  \ \  \ \ \ 
		<\overline{M}_0\overline{S}_0^2 L_2\left(\sigma+a\overline{S}_0\right)\triangleq \Delta_\mathcal{R},\\
		{{|\mathcal{W}(t)|=\left|\frac{2}{a}\nu(t)M_0(t-D)\right|\leq \frac{2}{a}\nu^*\overline{M}_0}},\\
		\left|\dot{\tilde{\theta}}(t)\right|\leq
		\left|k\mathcal{F}(t)\right|+|ka\mathcal{R}(t)| {+|k\mathcal{W}(t)|}
		<k\big(\Delta_\mathcal{F}+a\Delta_\mathcal{R} {+\frac{2}{a}\nu^*\overline{M}_0}\big)\triangleq \Delta_\theta.
	\end{array}
\end{equation}
From (\ref{2tvr_DefineG}), (\ref{2tvr_DefineY}), and \eqref{2tvr_errorSystemZYR},  we derive the following results
\begin{equation}\label{2tvr_boundsGY}
	\begin{array}{l}
		|G(t)|= \frac{1}{\varepsilon}\left|\int_{t-\varepsilon}^{t}(\tau-t+\varepsilon)\mathcal{F}(\tau){\rm{d}}\tau\right|
		\leq \frac{1}{\varepsilon}\int_{t-\varepsilon}^{t}(\tau-t+\varepsilon){\rm{d}}\tau\cdot|\mathcal{F}(\tau)|<\frac{\Delta_\mathcal{F}}{2}\varepsilon,\\
		|Y_1(t)|= \frac{1}{\varepsilon}\left|\int_{t-\varepsilon}^{t}\int_{\tau-h(\tau)}^{t-D}M_0(\tau-D)\frac{\partial q}{\partial \theta}\big(\theta^{*}+\tilde{\theta}(s)\big)\dot{\tilde{\theta}}(s){\rm{d}}s{\rm{d}}\tau\right|\\  \ \  \ \  \ \  \ \  \ \ \
		\leq \frac{1}{\varepsilon}\int_{t-\varepsilon}^{t}\int_{\tau-h(\tau)}^{t-D}{\rm{d}}s{\rm{d}}\tau \cdot \overline{M}_0 \left|\frac{\partial q}{\partial \theta}\big(\theta^{*}+\tilde{\theta}(s)\big)\right| \cdot \left|\dot{\tilde{\theta}}(s)\right|\\  \ \  \ \  \ \  \ \  \ \ \
		< \frac{1}{2}{\overline{M}_0q_1(\sigma)\Delta_\theta}(\varepsilon+2\bar{d}),\\
		|Y_2(t)|= \frac{1}{\varepsilon}\left|\int_{t-\varepsilon}^{t}\int_{\tau-h(\tau)}^{t-D}M_0(\tau-D)S^{\rm T}_0(\tau-h(\tau))\frac{\partial^2 q}{\partial \theta^2}\big(\theta^{*}+\tilde{\theta}(s)\big)\dot{\tilde{\theta}}(s){\rm{d}}s{\rm{d}}\tau
		\right|\\  \ \  \ \  \ \  \ \  \ \ \
		\leq \frac{1}{\varepsilon}\int_{t-\varepsilon}^{t}\int_{\tau-h(\tau)}^{t-D}{\rm{d}}s{\rm{d}}\tau \cdot \overline{M}_0\overline{S}_0 \left|\frac{\partial^2 q}{\partial \theta^2}\big(\theta^{*}+\tilde{\theta}(s)\big)\right| \cdot \left|\dot{\tilde{\theta}}(s)\right|\\  \ \  \ \  \ \  \ \  \ \ \
		< \frac{1}{2}{\overline{M}_0\overline{S}_0q_2(\sigma)\Delta_\theta}(\varepsilon+2\bar{d}),\\
		\left|\dot{z}(t)\right|\leq \left|k\frac{\partial q}{\partial \theta}^{\rm T}\big(\theta^{*}+z(t-D)+kG(t-D)\big)\right|  +\left|\frac{2k}{a}Y_1(t)\right|+\left|2kY_2(t)\right|+\left|ka\mathcal{R}(t)\right| {+|k\mathcal{W}(t)|}\\ \ \ \ \ \  \ \ \ \  
		<k\big(q_1(\sigma)+\frac{1}{a}{\overline{M}_0q_1(\sigma)\Delta_\theta}(\varepsilon+2\bar{d}) 
		+{\overline{M}_0\overline{S}_0q_2(\sigma)\Delta_\theta}(\varepsilon+2\bar{d})+a\Delta_\mathcal{R} {+\frac{2}{a}\nu^*\overline{M}_0}\big)
		\triangleq \Delta_z.
	\end{array}
\end{equation}
Furthermore, by applying the differential mean value theorem to \eqref{2tvr_define_Phi}, and under the bounds (\ref{2tvr_AssumptionBound2}) and (\ref{2tvr_boundsGY}), we arrive at
\begin{equation}\label{2tvr_boundphi}
	\begin{array}{l}
		|\Phi(t)|\leq \left|\frac{\partial^2 q}{\partial \theta^2}\big(\theta^{*}+z(t-D)+\zeta kG(t-D)\big)\right|\cdot |kG(t-D)|
		<q_2(\sigma)\frac{k\Delta_\mathcal{F}}{2}\varepsilon,
	\end{array}
\end{equation}
where $\zeta\in(0,1)$.	

\begin{theorem}\label{theorem2} Under Assumptions 3-4 with a given $\sigma$, consider the closed-loop system consisting of the multi-variable map (\ref{2tvr_map}) and the ES controller (\ref{2tvr_controller}), as well as the initial condition $\left|{\tilde{\theta}}(t)\right|\leq \sigma_0<\sigma$ for $t\in\left[0,\bar{h}\right]$. Given {$D\geq 0$} and tuning parameters $\sigma_0, k, a, \delta, \bar{d}^*, \varepsilon^*>0$, and {$\nu^*\geq 0$}, let scalar decision variables $\lambda_1, \lambda_2, \lambda_3, \lambda_4, \lambda_5, {\gamma
>0}$, and $P\geq
1, Q>0, R>0$ satisfy the following LMIs:
\begin{equation}\label{2tvr_lmi1}
	\begin{array}{l}
		\Omega_0=\begin{bmatrix}
			\Psi    &\Xi\\
			*       &-R
		\end{bmatrix}< 0,\ \ \ \  \Omega_1=\frac{\varpi}{2\delta}<\big(\sigma-\frac{k\Delta_\mathcal{F}}{2}\varepsilon^*\big)^2,\\
		\Omega_2=(P+{Q}D)\varkappa^2+RD{{}^3}\Delta_z^2<\big(\sigma-\frac{k\Delta_\mathcal{F}}{2}\varepsilon^*\big)^2,\\
		
	\end{array}
\end{equation}
where $\Psi$ is the symmetric matrix composed of
\begin{equation}
	\begin{array}{l}\label{2tvr_Psi}
		\Psi_{11}=-\big(2k\mu(\sigma)-2\delta\big) P+Q, \ \
		\Psi_{13}=kL_1P, \\
		\Psi_{15}=kP, \ \ \ \ \ \Psi_{16}=\frac{2k}{a}P,\ \ \Psi_{17}=2kP,
		\\
		\Psi_{18}=kaP,\ \ \
		{\Psi_{19}=kP,}  \ \ \ \ \
		\Psi_{22}=-Qe^{-2\delta D}+\lambda_1 L_1^2,\\
		\Psi_{33}=-Re^{-2\delta D},\ \
		\Psi_{44}=-\lambda_1,\ \
		\Psi_{55}=-\frac{\lambda_2}{\varepsilon^*},\\
		\Psi_{66}=-\frac{\lambda_3}{\varepsilon^*+2\bar{d}^*},\ \ \
		\Psi_{77}=-\frac{\lambda_4}{\varepsilon^*+2\bar{d}^*},\ \ \
		\Psi_{88}=-\lambda_5a, \ \ \  {\Psi_{99}=-\gamma,}
	\end{array}
\end{equation}
with other terms being zero, and
\begin{equation}\label{2tvr_Xi}
	\begin{array}{l}
		\Xi=\left[0,0,0,k,k,\frac{2k}{a},2k,ka {,k}\right]^{\rm T}DR,\\
		\varkappa=\sigma_0+\Delta_\theta \varepsilon^*+\frac{k \Delta_\mathcal{F}}{2}\varepsilon^*,\ \
		\\
		\varpi=\frac{\lambda_2q_2^2(\sigma)k^2 \Delta_\mathcal{F}^2}{4}\varepsilon^*+\frac{\lambda_3\overline{M}_0^2 q_1^2(\sigma)\Delta_\theta^2}{4}(\varepsilon^*+2\bar{d}^*)
		 +\frac{\lambda_4 \overline{M}_0^2\overline{S}_0^2q_2^2(\sigma)\Delta_\theta^2}{4}(\varepsilon^*+2\bar{d}^*)+\lambda_5a \Delta_\mathcal{R}^2 {+\frac{4\gamma\overline{M}_0^2}{a^2}{\nu^*}^2}.
	\end{array}
\end{equation}
Then, for $\forall$ $\bar{d} \in [0, \bar{d}^*]$ and $\forall$ $\varepsilon \in (2^{n+1}\bar{d}, \varepsilon^*]$, the estimation error satisfies
\begin{equation}\label{2tvr_theorem1boundtheta}
	\begin{array}{l}
		\Big|{\tilde{\theta}}(t)\Big|< \left|{\tilde{\theta}}(\bar{h})\right|+\Delta_\theta (t-\bar{h})<\sigma, \ \
		t\in \left[\bar{h}, \varepsilon+\bar{h}\right],\\
		\left|{\tilde{\theta}}(t)\right|<  \left[\big((P+{Q}D)\varkappa^2+RD^3\Delta_z^2\big)e^{-2\delta(t-\varepsilon-\bar{h})}\right.
		 +\left.\big(1-e^{-2\delta(t-\varepsilon-\bar{h})}\big)\frac{\varpi}{2\delta}\right]^{\frac{1}{2}} +\frac{k\Delta_\mathcal{F}}{2}\varepsilon<\sigma,\ \ \ t\in \left[\varepsilon+\bar{h}, \infty \right),
	\end{array}
\end{equation}
and is exponentially attracted to the set
\begin{equation}
	\begin{array}{l}
		\Theta=\left\{\tilde{\theta} \in \mathbb{R}^n: \left|{\tilde{\theta}}\right|< \sqrt{\frac{\varpi}{2\delta}}+\frac{k\Delta_\mathcal{F}}{2}\varepsilon \right\},
	\end{array}
\end{equation}
which is adjustable via the tuning parameters $\varepsilon$ and $a$. Moreover, LMIs \eqref{2tvr_lmi1} are always feasible for small enough $\bar{d}^*,\varepsilon^*,{\nu^*}$, $a$, and $k$. 
\end{theorem}
\textit{Proof}: The proof is given in Appendix.

\section{Examples}
\subsection{Scalar Case}
\begin{table}[t,width=.8\linewidth]\rmfamily
	\centering
	\caption{Tuning parameters: $k=0.002, \delta=0.0003, a=0.04, \sigma_0=1.6, \sigma=1.7, \mu(\sigma)=0.3$.}
	\label{tab:scalar}
	\renewcommand{\arraystretch}{1.1}
	\begin{tabularx}{0.8\linewidth}{XXXXXXX}
		\hline
		& $\nu^*$& $\bar{d}$ &$D$  & $\varepsilon$ & UB-1 & UB-2\\ \hline
		\multirow{2}{*}{\cite{gaofeng2023}} & \multirow{2}{*}{-}  & \multirow{2}{*}{-} &\multirow{2}{*}{-} & 0.0005 & 0.576& 0.652  \\ \cline{5-7}
		& &  &  & 0.001 & 0.687& 0.828  \\ \cline{1-7}
		\multirow{2}{*}{\cite{jianzhongcdc2024}}  
		& \multirow{2}{*}{-} & \multirow{2}{*}{-} &\multirow{2}{*}{0.5}  & 0.0005& 0.609 &0.751 \\ \cline{5-7}
		& &  &   & 0.001& 0.728 &0.956 \\ \cline{1-7}
		\multirow{4}{*}{Theorem 1} & \multirow{2}{*}{-} & \multirow{2}{*}{0.0001} &\multirow{2}{*}{1.0}  & 0.0005& 0.695 &0.999 \\ \cline{5-7}
		& &  &   & 0.001& 0.822 &1.246 \\ \cline{2-7}
		& \multirow{2}{*}{0.0002} & \multirow{2}{*}{0.0001} &\multirow{2}{*}{1.0}  & 0.0005& 0.732 &1.046 \\ \cline{5-7}
		& &  &   & 0.001& 0.859 &1.292 \\
		\hline
	\end{tabularx}
\end{table}
\begin{table}[t,width=.8\linewidth]\rmfamily
	\caption{Tuning parameters: $\varepsilon=0.0005, \bar{d}=0.0001, \nu^*=0.0002, a=0.04, \sigma_0=1.6, \sigma=1.7, \mu(\sigma)=0.3$.}
	\label{tab:scalar_2}
	\renewcommand{\arraystretch}{1.1}
	\begin{tabularx}{0.8\linewidth}{XXXXX}
		\hline
		& $k$ &$\delta$  & $D$ & UB-1 \\ \hline
		\multirow{3}{*}{Theorem 1}   & \multirow{1}{*}{0.001} &\multirow{1}{*}{0.00015} & 6.14 & 0.900  \\ \cline{2-5}
		& \multirow{1}{*}{0.002} &\multirow{1}{*}{0.0003}  & 2.42& 0.900  \\ \cline{2-5}
		& \multirow{1}{*}{0.004} &\multirow{1}{*}{0.0006}  &0.46  &0.900      \\ 
		\hline
	\end{tabularx}
\end{table}

Consider the nonlinear scalar map
\begin{equation}
	\begin{array}{l}
		q(\theta)=\frac{1}{3}\theta^3-\theta,
	\end{array}
	\nonumber
\end{equation}
and we employ the ES controller (\ref{1tvr_controller}) to locate the local maximum {$y^*=q(\theta^*)$} at $\theta^*=-1$, under transmission delays in the closed-loop system. Let us present two sets of data: For a given $\sigma$, the ``\textit{Real Bounds}'' refer to the exact bounds calculated by (\ref{1tvr_AssumptionBound2}) as if we know the map exactly, whereas the ``\textit{Estimated Bounds}'' refer to the approximate bounds which are somewhat larger than the real bounds. \textit{Real Bounds}: $q_0=3.861, q_1=6.290, q_2=5.400, L_1=3.730,  L_2=2.000$. \textit{Estimated Bounds}: $q_0=5.2297, q_1=7.410, q_2=5.800, L_1=5.880, L_2=2.050$. The solutions for both \textit{real bounds} and \textit{estimated bounds} via Theorem 1 are listed in Table \ref{tab:scalar}, where ``UB-1" and ``UB-2" denote the ultimate bound $\lim\limits_{t \to \infty}\sup\big|{\tilde{\theta}}(t)\big|$ obtained for the two data sets, respectively. 
The numerical simulations are revealed in {Fig. \ref{scalar_output_trajectory}, where the parameters were set as $a=0.04, \varepsilon=0.005, \bar{d}=1.25\times 10^{-3}, \nu^*=0.0002$.
	The results of Fig. \ref{scalar_output_trajectory} indicate that the system becomes increasingly unstable with larger $D$, and the controller gain can be reduced to partially compensate for the effects of significant time delays. This observation is further supported by the results presented in Table \ref{tab:scalar_2}.}

\begin{figure}[pos = H,abovecap=0pt]
	\centering
		\includegraphics[scale=0.5]{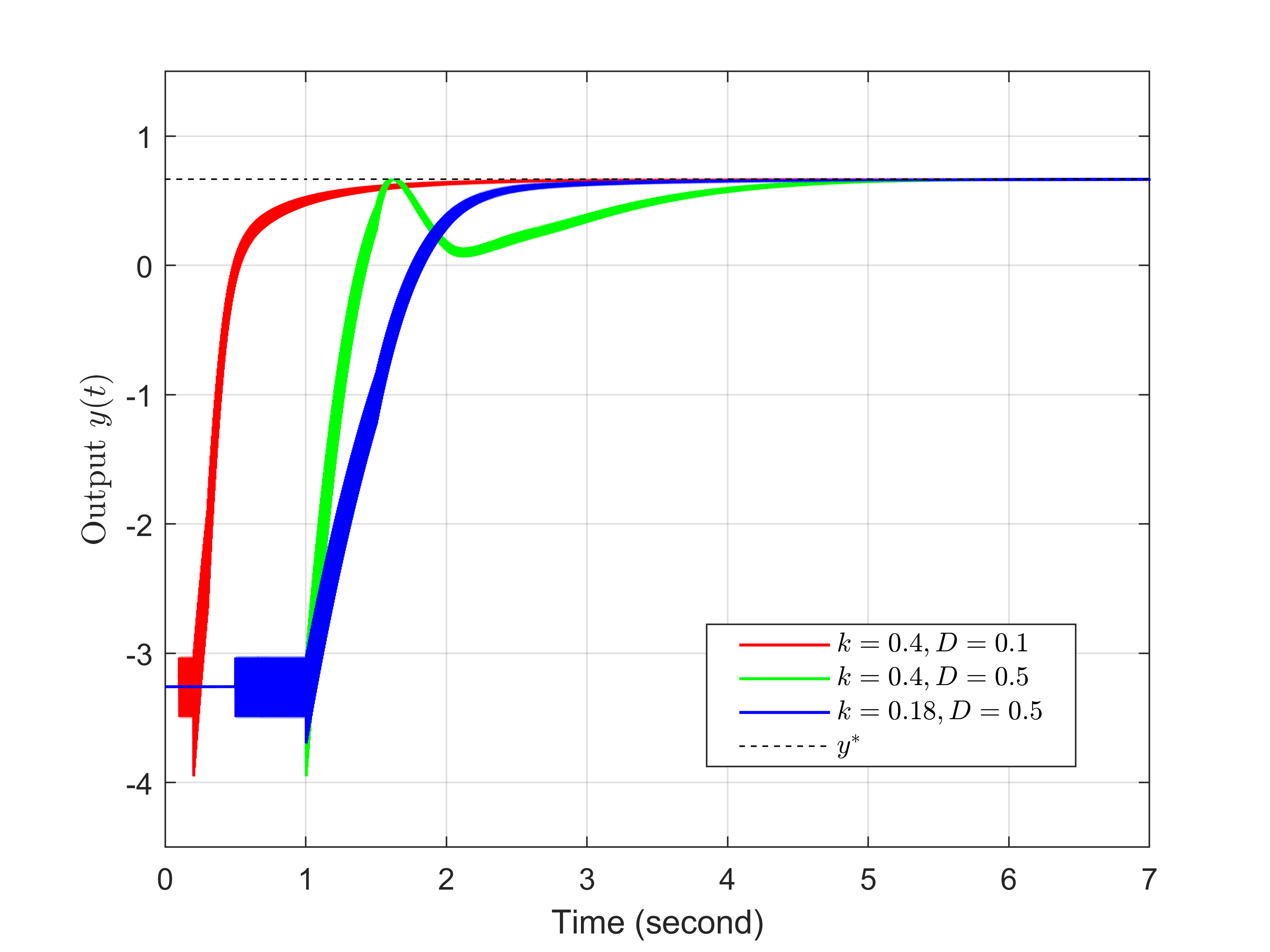}
		\caption{The trajectory of the output (scalar case).}
		\label{scalar_output_trajectory}
\end{figure}

\subsection{Vector Case}
\begin{table}[h,width=.8\linewidth]\rmfamily
	\caption{Tuning parameters: $k=0.004, \delta=0.002, a=0.06, \sigma_0=0.9, \sigma=1.0, \mu(\sigma)=1.0$.}
	\label{tab:vector}
	\renewcommand{\arraystretch}{1.1}
	\begin{tabularx}{0.8\linewidth}{XXXXXXX}
		\hline
		& $\nu^*$& $\bar{d}$ &$D$  & $\varepsilon$ & UB-1 & UB-2\\ \hline
		\multirow{2}{*}{\cite{gaofeng2023}}  & \multirow{2}{*}{-} & \multirow{2}{*}{-} &\multirow{2}{*}{-} & 0.01 & 0.264& 0.326  \\ \cline{5-7}
		&  & &  & 0.02 & 0.385& 0.468  \\ \cline{1-7}
		\multirow{2}{*}{\cite{jianzhongcdc2024}}   
		& \multirow{2}{*}{-} & \multirow{2}{*}{-} &\multirow{2}{*}{4.0}  & 0.01& 0.284 &0.374 \\ \cline{5-7}
		& &  &   & 0.02& 0.417 &0.541 \\ \cline{1-7}
		\multirow{4}{*}{Theorem 2}& \multirow{2}{*}{-} & \multirow{2}{*}{0.001} &\multirow{2}{*}{8.0}  & 0.01& 0.329 &0.468 \\ \cline{5-7}
		& &  &   & 0.02& 0.470 &0.659 \\ \cline{2-7}
		 & \multirow{2}{*}{0.001} & \multirow{2}{*}{0.001} &\multirow{2}{*}{8.0}  & 0.01& 0.382 &0.529 \\ \cline{5-7}
		& &  &   & 0.02& 0.523 &0.721 \\
		\hline
	\end{tabularx}
\end{table}
\begin{table}[h,width=.8\linewidth]\rmfamily
	 \centering
	\caption{Tuning parameters: $\varepsilon=0.01, \bar{d}=0.001, \nu^*=0.001, a=0.06, \sigma_0=0.9, \sigma=1.0, \mu(\sigma)=1.0$.}
	\label{tab:vector_2}
	\renewcommand{\arraystretch}{1.1}
	\begin{tabularx}{0.8\linewidth}{XXXXX}
		\hline
		& $k$ &$\delta$  & $D$ & UB-1 \\ \hline
		\multirow{3}{*}{Theorem 2}   & \multirow{1}{*}{0.002} &\multirow{1}{*}{0.001} & 53.2 & 0.420  \\ \cline{2-5}
		& \multirow{1}{*}{0.004} &\multirow{1}{*}{0.002}  & 14.0& 0.420  \\ \cline{2-5}
		& \multirow{1}{*}{0.006} &\multirow{1}{*}{0.003}  &0.52  &0.420      \\ 
		\hline
	\end{tabularx}
\end{table}
We consider a nonlinear function with two input variables
\begin{equation}
	\begin{array}{l}
		q(\theta_1, \theta_2)=-\frac{1}{2}(\theta_1^2+\theta_2^2)-\frac{1}{12}\theta_1^2\theta_2^2-\frac{1}{24}(\theta_1^4+\theta_2^4),
	\end{array}
	\nonumber
\end{equation}
which has a maximum at $(0, 0)$. By applying the ES regulator (\ref{2tvr_controller}), the extremum can be sought even in the presence of time delay. Here, two groups of data are also reported. \textit{Real Bounds}: $q_0=1.1668, q_1=1.8857, q_2=2.0, L_1=1.4033, L_2=0.7778$. \textit{Estimated Bounds}: $q_0=1.3, q_1=2.0, q_2=2.1, L_1=2.2, L_2=1.0$. Tables \ref{tab:vector}-\ref{tab:vector_2} show the LMI results attained by Theorem 2.
For the numerical solution, considering $K=[0.16,0;0,0.16]$, $a=0.02$, $\varepsilon=0.01$, $D=1.5$, $\bar {d}=0.0012, \nu^*=0.001$, $\hat{\theta}_1(0)=1$, $\hat{\theta}_2(0)=-1$, the trajectory is shown in Figs. \ref{simulation_dither}-\ref{trajectory_vector_3D}. Here the input signals $\theta_1(t),\theta_2(t)$ are measured at the exit of the controller.
\begin{figure}[pos = H,abovecap=0pt]
	\centering
		\includegraphics[scale=0.5]{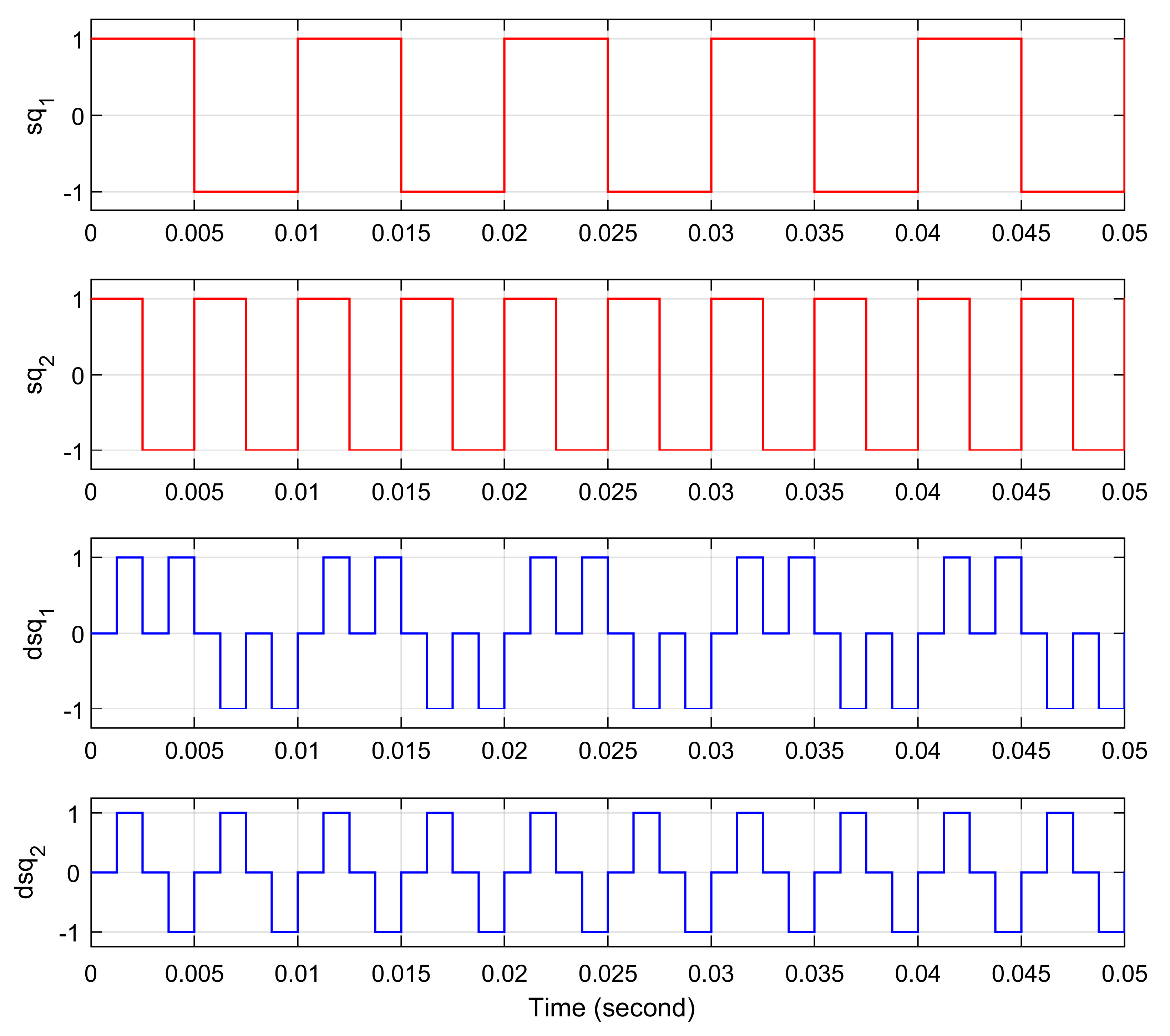}
		\caption{The square waves of the dithers.}
		\label{simulation_dither}
\end{figure}
\vspace{-1.5\baselineskip}
\begin{figure}[pos = H,abovecap=0pt]
	\centering
		\includegraphics[scale=0.55]{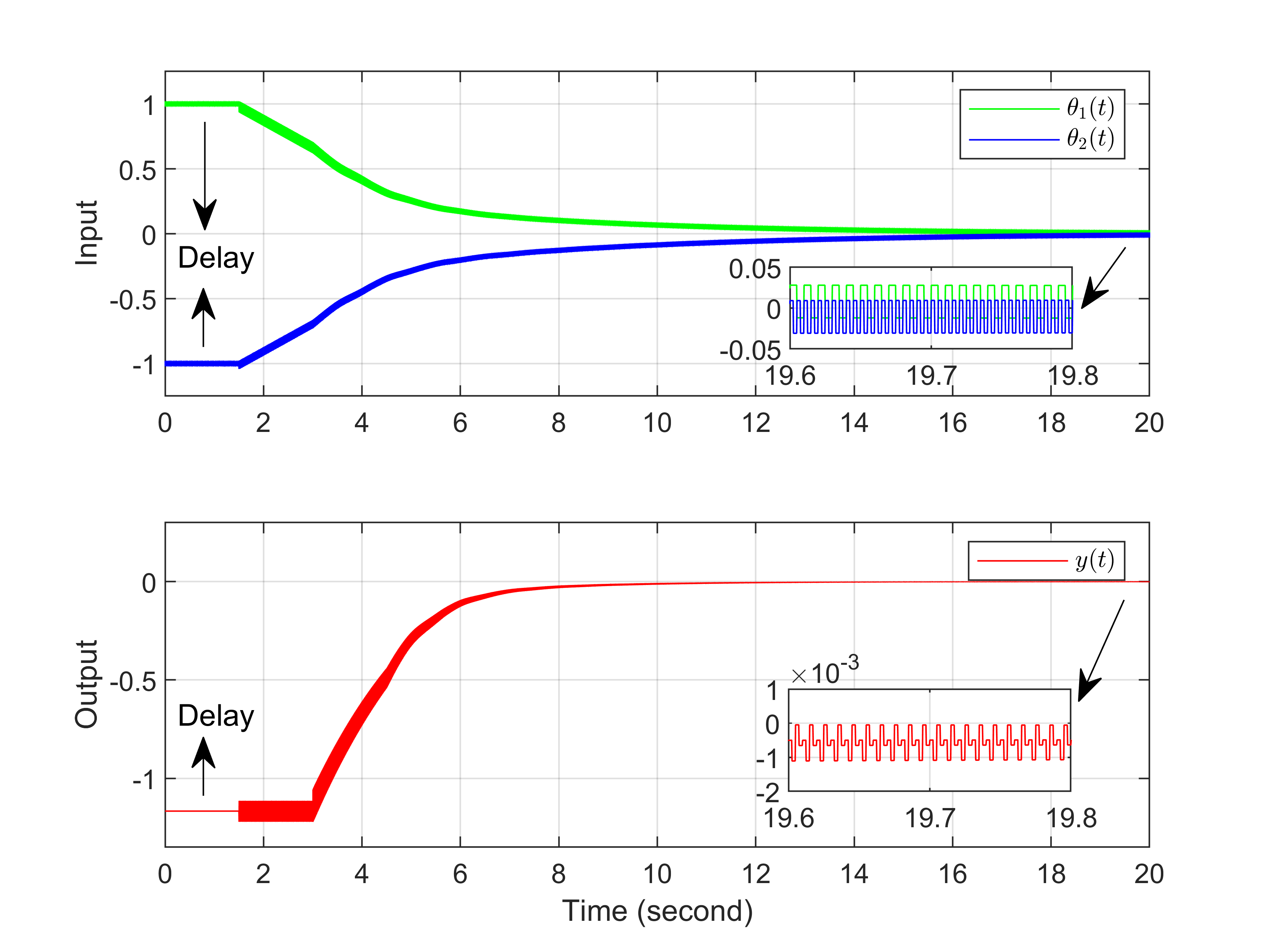}
		\caption{The trajectory of the input and output (vector case).}
		\label{trajectory_vector_2D}
\end{figure}
\vspace{-1.5\baselineskip}
\begin{figure}[pos=H,abovecap=0pt]
	\centering
		\includegraphics[scale=0.52]{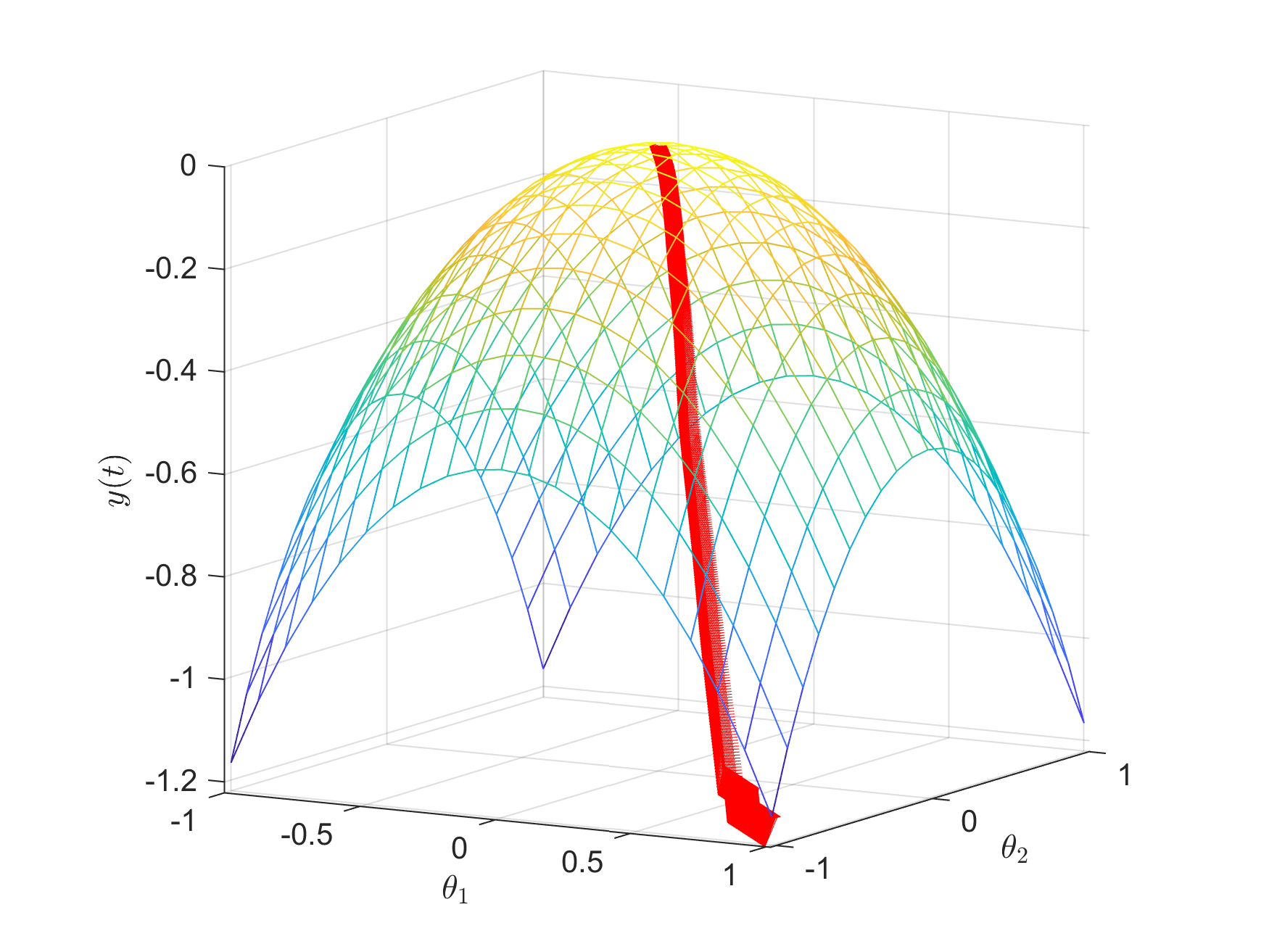}
		\caption{The three-dimension trajectory of ES.}
		\label{trajectory_vector_3D}
\end{figure}

\section{Conclusion}
In this paper, we develop a novel framework to manage with delay-robustness analysis for ES
of non-quadratic static maps {subject to a small disturbance}. The dither signals are carefully selected to handle the time-varying delay uncertainty. With the proposed Lyapunov functional and the resulting LMIs, we have offered quantitative upper bounds on the dither period and the delay that the ES system is able to tolerate. {Practical stability can be achieved by appropriately selecting ES parameters for any large known part of the constant delay.}

\appendix
\section*{Appendix. Proof of Theorem 2}
\setcounter{equation}{0}
\renewcommand\theequation{a.\arabic{equation}}

First of all, we prove the practical stability of \eqref{2tvr_errorSystemZYR}. By employing
\begin{equation}\label{2tvr_defineVp}
	V_P(t)=Pz^{\rm T}(t)z(t),
\end{equation}
and taking the time derivative of $V_P(t)$ along the trajectory of system \eqref{2tvr_errorSystemZYR}, we have
\begin{equation}\label{2tvr_dotVp}
	\begin{array}{l}
		\dot V_P(t)+2\delta V_P(t)=2Pz^{\rm T}(t)\dot z(t)+2\delta P z^{\rm T}(t)z(t)\\
		=2P z^{\rm T}(t)\Big[k\frac{\partial q}{\partial \theta}^{\rm T}\big(\theta^{*}+z(t-D)\big)+k\Phi(t)
		 -\frac{2k}{a}Y_1(t)-2kY_2(t)+ka\mathcal{R}(t)
		 {+k\mathcal{W}(t)}\Big]+2\delta P z^{\rm T}(t)z(t)\\
		=2P z^{\rm T}(t)\left[k\Phi(t)
		-\frac{2k}{a}Y_1(t)-2kY_2(t)+ka\mathcal{R}(t){+k\mathcal{W}(t)}\right]+2\delta P z^{\rm T}(t)z(t)\\
		\ \ 
		+2kPz^{\rm T}(t)\left[\frac{\partial q}{\partial \theta}^{\rm T}\big(\theta^{*}+z(t)\big)+\left(\frac{\partial q}{\partial \theta}^{\rm T}\big(\theta^{*}+z(t-D)\big)-\frac{\partial q}{\partial \theta}^{\rm T}\big(\theta^{*}+z(t)\big)\right)\right]\\
		\leq 2P z^{\rm T}(t)\left[k\Phi(t)
		-\frac{2k}{a}Y_1(t)-2kY_2(t)+ka\mathcal{R}(t){+k\mathcal{W}(t)}\right] +2\delta P z^{\rm T}(t)z(t)
		+2kPz^{\rm T}(t)\frac{\partial q}{\partial \theta}^{\rm T}\big(\theta^{*}+z(t)\big)\\
		\ \ \ +2kP\left|z(t)\right|\cdot L_1\cdot\left|z(t-D)-z(t)\right|\\
		\leq2 P|z(t)|\cdot\left[k|\Phi(t)|
		+\frac{2k}{a}|Y_1(t)|+2k|Y_2(t)|+ka|\mathcal{R}(t)|{+k|\mathcal{W}(t)|}\right]+2\delta P z^{\rm T}(t)z(t)-2kP\mu(\sigma)z^{\rm T}(t)z(t)\\
		\ \ \ +2kL_1P\left|z(t)\right| \cdot\left|\int_{t-D}^t \dot z(s){\rm{d}}s\right|,
	\end{array}
\end{equation}
where we have utilized the fact that $2kPz^{\rm T}(t)\frac{\partial q}{\partial \theta}^{\rm T}\big(\theta^{*}+z(t)\big)\leq -2kP\mu(\sigma)z^{\rm T}(t)z(t)$ from (\ref{2tvr_boundsAssumption}) in Assumption 3. Moreover, by taking
\begin{equation}\label{2tvr_defineVq1}
	\begin{array}{l}
		V_{Q}(t)=Q\int_{t-D}^t e^{2\delta(s-t)}z^{\rm T}(s)z(s){\rm{d}}s,
	\end{array}
\end{equation}
and the time derivative of \eqref{2tvr_defineVq1}, it can be derived that
\begin{equation}\label{2tvr_dotVq1}
	\begin{array}{l}
		\dot V_{Q}(t)+2\delta V_{Q}(t)=Qz^{\rm T}(t)z(t)-Qe^{-2\delta D}z^{\rm T}(t-D)z(t-D).
	\end{array}
\end{equation}
Besides, using
\begin{equation}\label{2tvr_defineVr}
	\begin{array}{l}
		V_R(t)=RD\int_{t-D}^t e^{2\delta(s-t)}(s-t+D)\dot z^{\rm T}(s)\dot z(s){\rm{d}}s,
	\end{array}
\end{equation}
and differentiating \eqref{2tvr_defineVr} with respect to $t$ leads to
\begin{equation}
	\begin{array}{l}\label{2tvr_dotVr_1}
		\dot V_R(t)+2\delta V_R(t)\\=RD^2\dot z^{\rm T}(t)\dot z(t)-RD\int_{t-D}^t e^{2\delta(s-t)}\dot z^{\rm T}(s)\dot z(s){\rm{d}}s\\
		\leq RD^2\dot z^{\rm T}(t)\dot z(t)-RDe^{-2\delta D}\int_{t-D}^t \dot z^{\rm T}(s)\dot z(s){\rm{d}}s\\
		\leq RD^2\dot z^{\rm T}(t)\dot z(t)-Re^{-2\delta D}\int_{t-D}^t \dot z^{\rm T}(s){\rm{d}}s\cdot \int_{t-D}^t \dot z(s){\rm{d}}s,
	\end{array}
\end{equation}
in which Jensen's inequality is employed. Under Assumptions 3--4, from $\frac{\partial q}{\partial \theta}\left(\theta^{*}\right)=0$ and $\big|\frac{\partial q}{\partial \theta}\left(\theta^{*}+z(t-D)\right)-\frac{\partial q}{\partial \theta}\left(\theta^{*}\right)\big|< L_1|z(t-D)|$, it follows that
\begin{equation}\label{2tvr_qq}
	\begin{array}{l}
		0< L_1^2|z(t-D)|^2-\big|\frac{\partial q}{\partial \theta}\big(\theta^{*}+z(t-D)\big)\big|^2 .
	\end{array}
\end{equation}

{Now based on \eqref{2tvr_defineVp}, \eqref{2tvr_defineVq1}, and \eqref{2tvr_defineVr}, we construct the following Lyapunov functional candidate:}
\begin{equation}\label{2tvr_defineV}
	\begin{array}{l}
		V(t)=V_P(t)+V_{Q}(t)+V_R(t),
	\end{array}
\end{equation}
{in which scalars $P\geq 1, Q>0, R>0$, are decision variables.}
According to (\ref{2tvr_dotVp}), (\ref{2tvr_dotVq1}), (\ref{2tvr_dotVr_1}), and (\ref{2tvr_qq}), and by employing the S-procedure with positive parameters $\lambda_i (i=1,2,3,4,5)$, {and $\gamma$}, for $t\geq \varepsilon+\bar{h}$, we obtain
\begin{equation}\label{2tvr_dV}
	\begin{array}{l}
		\dot V(t)+2\delta V(t)
		+\lambda _1\left[L_1^2|z(t-D)|^2-\big|\frac{\partial q}{\partial \theta}\big(\theta^{*}+z(t-D)\big)\big|^2 \right]\\
		-\frac{\lambda_2}{\varepsilon}\Phi^{\rm T}(t)\Phi(t)-\frac{\lambda_3}{\varepsilon+2\bar{d}}Y_1^{\rm T}(t)Y_1(t)-\frac{\lambda_4}{\varepsilon+2\bar{d}}Y_2^{\rm T}(t)Y_2(t)-\lambda_5a\mathcal{R}^{\rm T}(t)\mathcal{R}(t) {-\gamma\mathcal{W}^{\rm T}(t)\mathcal{W}(t)}\\
		\leq \eta ^{\rm T}(t)\Psi\eta(t)+RD^2\dot z^{\rm T}(t)\dot z(t)\\
		\leq \eta ^{\rm T}(t)\Psi\eta(t)+\eta ^{\rm T}(t)\Xi R^{-1}\Xi^{\rm T}\eta(t)<0,
	\end{array}
\end{equation}
where $\eta(t)=col\{|z(t)|,|z(t-D)|,\big|\int_{t-D}^t \dot z(s){\rm{d}}s\big|, \big|\frac{\partial q}{\partial \theta}\left(\theta^{*}+z(t-D)\right)\big|,|\Phi(t)|,  |Y_1(t)|,|Y_2(t)|, |\mathcal{R}(t)|{, |\mathcal{W}(t)|}\}$,
$\Psi$ and $\Xi$ are given in (\ref{2tvr_Psi}) and (\ref{2tvr_Xi}), respectively. The last inequality of (\ref{2tvr_dV}) follows from $\Omega_0<0$ given in (\ref{2tvr_lmi1}) in the sense of Schur complement.

Considering (\ref{2tvr_boundFR})--\eqref{2tvr_boundphi}, the inequality \eqref{2tvr_dV} suggests
\begin{equation}\label{2tvr_A8}
	\begin{array}{l}
		\dot V(t)+2\delta V(t)<
		\frac{\lambda_2}{\varepsilon}\Phi^{\rm T}(t)\Phi(t)+\frac{\lambda_3}{\varepsilon+2\bar{d}}Y_1^{\rm T}(t)Y_1(t)+\frac{\lambda_4}{\varepsilon+2\bar{d}}Y_2^{\rm T}(t)Y_2(t) +\lambda_5a\mathcal{R}^{\rm T}(t)\mathcal{R}(t)
		{+\gamma\mathcal{W}^{\rm T}(t)\mathcal{W}(t)}\\
		< \frac{\lambda_2q_2^2(\sigma)k^2 \Delta_\mathcal{F}^2}{4}\varepsilon+\frac{\lambda_3\overline{M}_0^2 q_1^2(\sigma)\Delta_\theta^2}{4}(\varepsilon+2\bar{d})
		 +\frac{\lambda_4 \overline{M}_0^2\overline{S}_0^2q_2^2(\sigma)\Delta_\theta^2}{4}(\varepsilon+2\bar{d})+\lambda_5a \Delta_\mathcal{R}^2 {+\frac{4\gamma\overline{M}_0^2}{a^2}{\nu^*}^2}\triangleq\varpi.
	\end{array}
\end{equation}
By resorting to the comparison principle for (\ref{2tvr_A8}), for $t\geq \varepsilon+\bar{h}$, we have 
\begin{equation}\label{2tvr_V<M}
	\begin{array}{l}
		V(t)\leq V(\varepsilon+\bar{h})e^{-2\delta(t-\varepsilon-\bar{h})} +\big(1-e^{-2\delta(t-\varepsilon-\bar{h})}\big)\frac{\varpi}{2\delta}\triangleq N(t).
	\end{array}
\end{equation}
On the other hand, based on \eqref{2tvr_defineVp} and \eqref{2tvr_defineV}, under the condition $P\geq 1$, it follows that
\begin{equation}\label{2tvr_boundV}
	\begin{array}{l}
		V(t)\geq Pz^{\rm T}(t)z(t) \geq \left|z(t)\right|^2.
	\end{array}
\end{equation}

Next, we show the practical stability of \eqref{2tvr_errorSystemFR} and the supposed overall bound \eqref{2tvr_overallBound} is not violated. Under the initial condition underneath \eqref{2tvr_define_esterror}, i.e., $\left|\tilde{\theta}(t)\right|\le\sigma_0<\sigma$ for all $t\in\left[0,\bar{h}\right]$, considering (\ref{2tvr_boundFR}), we get
\begin{equation}\label{2tvr_theta1}
	\begin{array}{l}
		\left|{\tilde{\theta}}(t)\right|=\left|{\tilde{\theta}}\left(\bar{h}\right)+\int_{\bar{h}}^t\dot{\tilde{\theta}}(\tau){\rm{d}}\tau \right|< \left|{\tilde{\theta}}\left(\bar{h}\right)\right|+\Delta_\theta \left(t-\bar{h}\right)
		\le\sigma_0+\Delta_\theta\varepsilon, \ \
		t\in \left[\bar{h}, \varepsilon+\bar{h}\right].
	\end{array}
\end{equation}
Equation (\ref{2tvr_theta1}) corresponds to the 1st formula in (\ref{2tvr_theorem1boundtheta}) and is ensured by $\Omega_2<\big(\sigma-\frac{k\Delta_\mathcal{F}}{2}\varepsilon\big)^2$ in (\ref{2tvr_lmi1}). Furthermore, from (\ref{2tvr_denoteZ}) and \eqref{2tvr_boundV}, we conclude that
\begin{equation}\label{2tvr_theta<U}
	\begin{array}{l}
		\big|{\tilde{\theta}}(t)\big|\leq \left|z(t)\right|+\left|kG(t)\right|\leq\sqrt{V(t)}+\left|kG(t)\right|< \sqrt{N(t)}+\frac{k \Delta_\mathcal{F}}{2}\varepsilon, \ \ t\geq \varepsilon+\bar{h},
	\end{array}
\end{equation}
which corresponds to the 2nd formula in (\ref{2tvr_theorem1boundtheta}). Notice that $N(t)$ defined in (\ref{2tvr_V<M}) is monotonically varying with respect to $t$. Therefore, from (\ref{2tvr_theta<U}), to guarantee that
\begin{equation*}
	\begin{array}{l}
		\big|{\tilde{\theta}}(t)\big|<\sqrt{N(t)}+\frac{k \Delta_\mathcal{F}}{2}\varepsilon<\sigma, \ \ t \in [\varepsilon+\bar{h}, \infty),
	\end{array}
\end{equation*}
the inequality must be satisfied at the two boundary points, namely,
\begin{equation}\label{2tvr_theta38}
	\Bigg\{\begin{array}{ll}\big|{\tilde{\theta}}(\varepsilon+\bar{h})\big|<\sqrt{V(\varepsilon+\bar{h})}+\frac{k \Delta_\mathcal{F}}{2}\varepsilon<\sigma,\\
		\lim\limits_{t \to \infty}\big|{\tilde{\theta}}(t)\big|<\lim\limits_{t \to \infty}\sqrt{N(t)}+\frac{k \Delta_\mathcal{F}}{2}\varepsilon<\sigma.
	\end{array}
\end{equation}
It can be derived from (\ref{2tvr_denoteZ}) and \eqref{2tvr_theta1} that
\begin{equation*}
	\begin{array}{l}
		\left|z(t)\right|\leq \big|{\tilde{\theta}}(t)\big|+\left|kG(t)\right|<
		\sigma_0+\Delta_\theta (t-\bar{h})+\frac{k \Delta_\mathcal{F}}{2}\varepsilon,
	\end{array}
\end{equation*}
and then, we further attain
\begin{equation*}
	\begin{array}{l}
		\left|z(\varepsilon+\bar{h})\right|<
		\sigma_0+\Delta_\theta \varepsilon+\frac{k \Delta_\mathcal{F}}{2}\varepsilon\triangleq\varkappa.
	\end{array}
\end{equation*}
Thus, based on (\ref{2tvr_defineVp}), we have
\begin{equation}\label{2tvr_boundVp}
	\begin{array}{l}
		V_P(\varepsilon+\bar{h})=Pz^{\rm T}(\varepsilon+\bar{h})z(\varepsilon+\bar{h})=P\left|z(\varepsilon+\bar{h})\right|^2
		<P\varkappa^2.
	\end{array}
\end{equation}
It follows from (\ref{2tvr_defineVq1}) that
\begin{equation}\label{2tvr_boundVq}
	\begin{array}{l}
		V_{Q}(\varepsilon+\bar{h})=Q\int_{\varepsilon+\bar{h}-D}^{\varepsilon+\bar{h}} e^{2\delta(s-\varepsilon-\bar{h})}z^{\rm T}(s)z(s){\rm{d}}s
		\leq Q\int_{\varepsilon+\bar{h}-D}^{\varepsilon+\bar{h}} \left|z(s)\right|^2{\rm{d}}s
		<QD\varkappa^2.
	\end{array}
\end{equation}
From (\ref{2tvr_defineVr}) and \eqref{2tvr_boundsGY}, we get
\begin{equation}\label{2tvr_boundVr}
	\begin{array}{l}
		V_R(\varepsilon+\bar{h})=	RD\int_{\varepsilon+\bar{h}-D}^{\varepsilon+\bar{h}} e^{2\delta(s-\varepsilon-\bar{h})}(s-\varepsilon-\bar{h}+D)\dot z^{\rm T}(s)\dot z(s){\rm{d}}s\\
		<RD^2\int_{\varepsilon+\bar{h}-D}^{\varepsilon+\bar{h}} \dot z^{\rm T}(s)\dot z(s){\rm{d}}s
		=RD^2\int_{\varepsilon+\bar{h}-D}^{\varepsilon+\bar{h}} \left|\dot z(s)\right|^2{\rm{d}}s
		<RD^3\Delta_z^2.
	\end{array}
\end{equation}
In order to let $z(t)$ in (\ref{2tvr_boundVq}) and $\dot{z}(t)$ in (\ref{2tvr_boundVr}) be well-defined at $t\in[\varepsilon+\bar{h}-D,\varepsilon+\bar{h}]$, we supplement the initial condition such that $\big|\hat{\theta}(t)-\theta^*\big|\le\sigma_0<\sigma$ for $t\in\left[-\bar{h},0\right]$, so that $\big|\tilde{\theta}(t)\big|\le\sigma_0<\sigma$ for all $t\in[-\bar{h},0]$. Consequently, considering (\ref{2tvr_defineV}) and (\ref{2tvr_boundVp})--(\ref{2tvr_boundVr}), we arrive at
\begin{equation}\label{2tvr_V3}
	\begin{array}{l}
		V(\varepsilon+\bar{h})<(P+QD)\varkappa^2+RD^3\Delta_z^2.
	\end{array}
\end{equation}
Substituting (\ref{2tvr_V3}) into (\ref{2tvr_theta38}), the 1st formula of (\ref{2tvr_theta38}) is ensured  by $\Omega_2<\big(\sigma-\frac{k\Delta_\mathcal{F}}{2}\varepsilon\big)^2$ in (\ref{2tvr_lmi1}), whereas the 2nd formula of (\ref{2tvr_theta38}) is guaranteed by $\Omega_1<\big(\sigma-\frac{k\Delta_\mathcal{F}}{2}\varepsilon\big)^2$ in (\ref{2tvr_lmi1}).
Therefore, through a contradiction argument similar to Appendix of \cite{ZHU2022109965,gaofeng2023}, it is proved that (\ref{2tvr_lmi1}) implies (\ref{2tvr_overallBound}).

{In the sequel, we verify the feasibility of \eqref{2tvr_lmi1} for given $\sigma$ and $D$, and small enough $\bar{d}^*,\varepsilon^*, {\nu^*}, a$, and $k$. First, choosing {large enough $\gamma>0$,}  $\lambda_i=1, i=2,3,4,5, a={(\varepsilon^*+2\bar{d}^*)}^{\frac{1}{3}}, \bar{d}^*\rightarrow 0^+, \varepsilon^*\rightarrow 0^+$, inequality $\Omega_0<0$ in \eqref{2tvr_lmi1} is feasible if the following inequality holds
	\begin{equation}\label{omega2}
			\begin{array}{l}
					\begin{array}{l}
						\begin{bmatrix}
							-(2k\mu(\sigma)-2\delta)P+Q&0&kL_1P&0\\
							*&-Qe^{-2\delta D}+\lambda_1 L_1^2&0&0\\
							*&*&-Re^{-2\delta D}&0\\
							*&*&*&-\lambda_1 
						\end{bmatrix}<0.
					\end{array}
			\end{array} 
	\end{equation}
	After that, by choosing $P=1, Q<2k\mu(\sigma)-2\delta, 0<\lambda_1<\frac{Q}{L_1^2}e^{-2\delta D}$, and large enough $R>0$, LMI \eqref{omega2} is always feasible. Note that the exponential decay rate $\delta>0$ is governed by the controller gain $k$ such that $2k\mu(\sigma)-2\delta>0$. 
	In addition, choosing $P=1, \lambda_i=1, i=2,3,4,5, a={(\varepsilon^*+2\bar{d}^*)}^{\frac{1}{3}}, \bar{d}^*\rightarrow 0^+, \varepsilon^*\rightarrow 0^+ {,\nu^*\rightarrow 0}$, and sufficiently small $Q,k$, it is evident that the inequalities $\Omega_1<\big(\sigma-\frac{k\Delta_\mathcal{F}}{2}\varepsilon^*\big)^2$ and $
	\Omega_2<\big(\sigma-\frac{k\Delta_\mathcal{F}}{2}\varepsilon^*\big)^2$ in \eqref{2tvr_lmi1} are feasible.} This completes the proof of Theorem 2. 

\printcredits

\section*{Declaration of competing interest}
The authors declare that they have no known competing financial interests or personal relationships that could have appeared to influence the work reported in this paper.

\bibliographystyle{elsarticle-num}

\bibliography{tvr_delay_refer_1}




%

\end{document}